%% file: Submitted_hybridization_in_crypto.tex
\documentclass[preprint,10pt]{elsarticle}

\input{preamble}
\journal{Computers \& Security}

\begin{document}

\begin{frontmatter}

\title{A Stackelberg Model for Hybridization in Cryptography}

\author[inst1]{Willie Kouam}
\author[inst1]{Stefan Rass}
\author[inst2]{Zahra Seyedi}
\author[inst1]{Shahzad Ahmad}
\author[inst3]{Eckhard Pfluegel}

\affiliation[inst1]{organization={Johannes Kepler University},
addressline = {LIT Secure and Correct Systems Lab},
city={Linz},
country={Austria}}

\affiliation[inst2]{organization={Polytechnic University of Milan},
addressline={Department of Electronics, Information and Bioengineering},
city={Milan},
country={Italy}}

\affiliation[inst3]{organization={Kingston University},
city={London Area},
country={United Kingdom}}

\begin{abstract}
Similar to a strategic interaction between rational and intelligent agents, cryptography problems can be examined through the prism of game theory. In this setting, the agent aiming to protect a message is called the defender, while the one attempting to decrypt it, generally for malicious purposes, is the attacker. To strengthen security in cryptography, various strategies have been developed, among which hybridization stands out as a key concept in modern cryptographic design. This strategy allows the defender to select among different encryption algorithms (classical, post-quantum, or hybrid) while carefully balancing security and operational costs. On the other side, the attacker, limited by available resources, chooses cryptanalysis methods capable of breaching the selected algorithm. We model this interaction as a Stackelberg cryptographic hybridization problem under resource constraints. Here, the defender randomizes over encryption algorithms, and the attacker observes the choice before selecting suitable cryptanalysis methods. The attacker's decision is framed as a conditional optimization problem, which we refer to as the ``attacker subgame''. We then propose a dynamic programming approach for the attacker’s subgame, while the defender's Stackelberg optimization is formulated as a linear program.
\end{abstract}


\end{frontmatter}

\section{Introduction}
Modern cryptographic deployments face a fundamental dilemma. The looming threat of quantum computing renders RSA and elliptic curve cryptography vulnerable to polynomial-time attacks \cite{shor1997}, while NIST-standardized post-quantum alternatives impose 2-5$\times$ computational overhead and 5-10 MB memory footprints \cite{alagic2022status}. No single algorithm satisfies the heterogeneous requirements of contemporary systems; TLS handshakes demand sub-100ms latency, IoT devices face severe memory constraints, and regulatory mandates require 20-year security guarantees. Consider, for instance, a financial institution processing transactions. RSA-2048 provides 10-20ms encryption \cite{gueron2002enhanced,bernstein2011ebacs} but faces quantum threats \cite{shor1997}. Lattice-based schemes offer quantum resistance at 50-100ms cost \cite{bos2018crystals,ducas2018crystals}. AES-256 is fast (submillisecond per block \cite{gueron2009intel}) and quantum-resistant \cite{chen2016report}  but requires secure key exchange. Traditional approaches treat algorithm selection as static cost minimization \cite{zhang2016case}: deploy the cheapest algorithm meeting a security threshold. These frameworks have critical flaws.

\begin{itemize}
    \item Security optimization treats adversaries as passive, ignoring their adaptive responses. Attackers observe deployed algorithms and optimally allocate resources among cryptanalytic methods, concentrating on RSA, which enables specialized factorization attacks, while diversification forces resource spreading.
    
    \item Real systems consume incommensurable resources (CPU, memory, bandwidth, latency). Single-cost aggregation obscures trade-offs: an algorithm cheap in CPU may be expensive in memory.
    
    \item Organizations rarely know adversary capabilities reliably. Nation-state attackers differ vastly from criminal organizations, and harvest-now-decrypt-later attacks create deep uncertainty about effective attack budgets.
    
    \item Regulatory standards mandate crypto-agility \cite{barker2018transitioning}, yet traditional optimization provides no principled diversification mechanism. Over-concentration creates single points of failure.
\end{itemize}

To address these limitations, we model cryptographic hybridization, strategic selection of algorithm portfolios, as a Stackelberg game. The defender (leader) acts first and commits to a mixed strategy $\boldsymbol{p} = (p_1, \ldots, p_n)$ over $n$ algorithms, where $p_i$ is the probability of deploying algorithm~$E_i$. The attacker (follower) observes the realized algorithm and subsequently selects a subset $S$ of cryptanalysis methods optimizing the probability of success as well as the cost to perform such an action, subject to a budget $k$. Knowing that the attacker will adjust their strategy in response, the defender anticipates these optimal moves and selects $\boldsymbol{p}$ to minimize the expected loss. In doing so, the defender must consider multiple practical constraints, such as CPU and memory usage, latency, quantum resistance, and diversification across different classes of encryption methods. This formulation captures essential strategic elements: observable defender commitments (via TLS negotiations, protocol parameters), attacker resource allocation under scarcity (attempting multiple methods increases success but consumes budget), and mixed-strategy uncertainty (forcing attackers to prepare for multiple contingencies). Our main contributions are summarized as follows:
\begin{itemize}
    \item  We provide the first Stackelberg formulation of cryptographic hybridization, unifying security games and cryptographic algorithm selection.
    \item We prove the attacker problem is non-monotone submodular maximization with knapsack constraints and propose a dynamic programming (DP) algorithm to solve it.
    \item  We formulate the defender decision-making problem as a linear program (LP) with heterogeneous resources (CPU, memory, latency), strategic requirements (quantum resilience), and diversification policies. We prove equilibrium existence and bound support size: with $u$ binding constraints, at most $u$ algorithms receive positive probability. 
    \item Assuming that the attacker’s available budget is not publicly known, we adopt a worst-case regret minimization approach to propose a robust solution under this uncertainty.
\end{itemize}
The remainder of this paper is organized as follows. Section~\ref{sec:related_work} reviews related work and positions our contribution within the existing literature. Section~\ref{sec:system_model} formalizes the system model, which serves as the basis for the subsequent analysis, and Section~\ref{sec:problem_formulation} is dedicated to the problem formulation. With this groundwork established, Section~\ref{sec:game_solution} develops the Stackelberg game solution with attacker and defender optimization.  Section~\ref{sec:worst_case_regret_minimization} implements the worst-case regret minimization approach. Section~\ref{sec:performance_analysis} subsequently discusses implementation details and performance analysis results, and finally, Section~\ref{sec:conclusion} concludes the paper.

\section{Related work} \label{sec:related_work}

Security games model strategic interactions between defenders and attackers. Stackelberg security games, where defenders commit first, and attackers observe before responding, have been successfully applied to airport security \cite{pita2008deployed}, patrolling strategies \cite{tambe2011security}, cyber-insurance~\cite{li2025insure}, and network defense scenarios \cite{zhang2023security, rass2023game, mvah2024countering}. The FlipIt game \cite{van2013flipit} which captures the strategic timing of actions, can be leveraged to model cryptographic key rotation, highlighting how game-theoretic approaches are relevant to real-world cryptographic practices. However, no prior work addresses strategic cryptographic algorithm selection under heterogeneous resource constraints and adaptive adversaries. In the domain of cryptographic optimization, algorithm selection has traditionally focused on single-objective problems: minimizing latency \cite{bernstein2011ebacs}, minimizing energy consumption, or minimizing cost subject to fixed security thresholds. Post-quantum migration strategies \cite{moody2016post} consider transition planning but treat security as a static parameter rather than the outcome of strategic interaction. These approaches aggregate heterogeneous resources (CPU, memory, bandwidth, latency) into scalar cost metrics, obscuring fundamental trade-offs between incommensurable dimensions. In this work, we approach the problem by recognizing that the choice of encryption and cryptanalysis algorithms is a strategic decision. Each party makes its selection while anticipating the likely actions of the other. This interaction occurs under budgetary constraints that reflect the practical parameters involved in evaluating the different algorithms.

Our attacker problem connects to the well-studied area of submodular maximization under knapsack constraints in combinatorial optimization. Ene and Nguyen~\cite{ene2017nearly} provided nearly linear-time algorithms with constant-factor approximation ratios for this problem class. Amanatidis et al.~\cite{amanatidis2021} develop sample-based and adaptive greedy methods that balance approximation quality with query complexity, which is crucial in our setting, where each "query" corresponds to evaluating an attack method's marginal contribution. Cui et al.~\cite{cui2023practical} achieve improved adaptivity bounds through parallel algorithms. However, most existing algorithms for these types of problems have been designed for large-scale instance problems. For instance, in the video recommendation study by Amanatidis et al.~\cite{amanatidis2020_fast}, the ground set includes $n = 62,000$ movies. In contrast, in our context, for any given encryption algorithm~$E_i$, the set of relevant cryptanalysis methods $\mathcal{D}(i)$ is usually much smaller. In practice, the cryptanalysis of a specific encryption algorithm~$E_i$ is performed by considering a finite set of known attack techniques, each corresponding to a distinct cryptanalytic paradigm such as differential, linear, integral, boomerang/rectangle, related‑key, algebraic, meet‑in‑the‑middle, impossible differential, and other variants, as systematically enumerated~\cite{lim2022sok}. This bounded set of methods typically falls into the tens rather than of thousands or more potential attacks. To handle this, we propose a DP algorithm for these more manageable instances, and if for a given algorithm~$E_i$ the value $|\mathcal{D}(i)|$ is not manageable by the DP algorithm, we instead employ the \textsc{SampleGreedy} algorithm~\cite{amanatidis2020_fast}, while tolerating the approximation error introduced.

The attacker’s available budget for carrying out actions may also be unknown to the defender. In such cases, robust optimization under uncertainty has been widely explored in areas like operational planning \cite{bernstein2011ebacs} and discrete optimization \cite{kouvelis2013robust}. Approaches based on minimax regret, as used in~\cite{ma2023optimizing}, aim to minimize worst-case performance degradation relative to optimal hindsight decisions. Although robust optimization has been used for security resource allocation, applying it to the selection of cryptographic algorithms, when the attacker’s budget is unknown to the defender, and the adversary acts strategically, is a novel approach. Our work fills a critical gap at the intersection of security games and applied cryptography by providing the first integration of Stackelberg games, multi-constraint optimization, and submodular maximization for cryptographic hybridization. We explicitly model defender-attacker strategic interaction while accounting for heterogeneous resources, quantum resilience requirements, diversification policies, and uncertainty in adversary capabilities.

\section{System model} \label{sec:system_model}

We model cryptographic algorithm selection as a two-player Stackelberg game between a defender (organization deploying cryptographic systems) and an attacker (adversary attempting to compromise encrypted data). This section formalizes the strategic scenario, players' action spaces, information structure, and payoff functions. The following conventions are employed in the subsequent section: sets are denoted by capital letters (e.g., $S, T$), vectors are represented in bold (e.g., $\boldsymbol{v}, \boldsymbol{w}$), scalars are written in lowercase (e.g., $x, y$), and the action sets of the players are indicated in cursive script (e.g., $\mathcal{A}, \mathcal{B}$). $E_i$ denotes a defender's encryption algorithm while $A_j$ represents an attacking method.

\subsection{Cryptographic deployment scenario}

Consider an organization tasked with choosing encryption algorithms to safeguard sensitive information, the latter faces many options. Modern cryptographic landscape provides a wide variety of options, including classical schemes such as RSA and elliptic curve cryptography, symmetric algorithms like AES and ChaCha20, post-quantum alternatives based on lattices or hashes, and hybrid approaches that combine multiple primitives. Each algorithm presents distinct trade-offs in computing cost, memory use, processing speed, resistance to quantum attacks, and cryptanalytic maturity (i.e., the extent to which it has been tested by cryptanalysis). The defender does not commit to a single algorithm but rather deploys a probability distribution over available algorithms. This may be realized through randomized selection at deployment time, heterogeneous deployment across different services or user sessions, or temporal rotation of algorithms. The attacker observes the algorithm protecting a particular target (through protocol negotiation, traffic analysis, or cryptanalytic reconnaissance) and must allocate limited computational resources among available cryptanalytic methods.

\subsection{Players' action spaces and Payoffs}

Let $\mathcal{A} = \{E_1, \ldots, E_n\}$ be the set of encryption algorithms available to the defender, and given an encryption algorithm~$E_i$ let $\mathcal{D}(i)$ denote the set of cryptanalytic methods that an attacker could apply to compromise it; the overall pool of attack methods is thus $\mathcal{D}_{\mathrm{all}} = \displaystyle \bigcup_{i=1}^{n} \mathcal{D}(i)$. The interaction between the parties is modeled as a sequential two‑player game in which the defender, acting as the leader, first selects an encryption algorithm from the set $\mathcal{A}$. Once this choice has been made and observed by the attacker, the latter acts as the follower and selects cryptanalytic methods from the corresponding set $\mathcal{D}(i)$ associated with the chosen encryption algorithm~$E_i$. The players' strategies and the game's outcome are described in the following.

\begin{itemize}
    \item Defender: A \emph{pure strategy} of the defender consists of deterministically deploying a single encryption algorithm~$E_i$. Conversely, a \emph{mixed strategy} is a probability distribution over the set of pure strategies, i.e., a function $f: \mathcal{A} \to [0,1] \text{ such that } \displaystyle \sum_{E_i \in \mathcal{A}} f(E_i) = 1.$ In our model, the defender publicly commits to a \emph{mixed strategy} $\boldsymbol{q} \in \Delta_n = \Big\{\, \boldsymbol{q} = (q_1, \ldots, q_n) \in \mathbb{R}^n \;\big|\; q_i \geqslant  0 \;\forall i, \ \displaystyle\sum_{i=1}^{n} q_i = 1 \,\Big\}$, where $q_i \equiv \Pr$(\text{Algorithm } $E_i$ is selected) denotes the probability that the encryption algorithm~$E_i$ is implemented.

    \item Attacker: Once the encryption algorithm~$E_i \in \mathcal{A}$ is revealed, the attacker selects a subset $S \subseteq \mathcal{D}(i)$ of appropriate cryptanalysis methods. Each method $A_j \in S$ consumes a specific amount of resources, represented by a cost $k_j$, and the attacker operates under a total resource budget~$k$, which implies that the condition $ \displaystyle  \sum_{A_j \in S} k_j \leqslant k$ must be satisfied. For every such subset $S$, the attacker’s cost function is defined so that its negative is submodular. A  function $h: 2^N \to \mathbb{R}$ is \emph{submodular} if 
    \begin{align*}
        h(A \cup \{j\}) - h(A) \geqslant h(B \cup \{j\}) - h(B), \quad \forall A \subseteq B \subset N, j \notin B.
    \end{align*}
\vspace{-0.6cm}
\end{itemize}\nobreak
\noindent For each pair of encryption and cryptanalysis algorithms $(E_i, A_j)$, let $s_{i,j} \in [0,1]$ the \emph{success probability} of the method~$A_j$ against algorithm~$E_i$ (i.e., the likelihood that the method $A_j$ successfully breaks algorithm~$E_i$); the resulting payoff structure is summarized in the normal-form representation in Table~\ref{normal_form_representation}. Each cell contains the tuple $\left((s_{i,j},\, c_i),\,(s_{i,j},\, k_j)\right)$ which corresponds to the pure strategy pair $(E_i, A_j)$ capturing the defender's \emph{(success probability, implementation cost)} and the attacker's \emph{(success probability, attack cost)} values, respectively. In our framework, the defender’s cost is vector-valued, capturing multiple practical criteria at once, including operational overhead, CPU and memory usage, latency, and reliability. Formally, this is represented as $c_i = \left( c_i^{(\mathrm{OP})}, c_i^{(\text{CPU})}, c_i^{(\text{MEM})}, \tau_i, r_i \right)$ and characterizes the trade-offs associated with choosing an encryption algorithm $E_i$: the fixed operational cost $c_i^{(\mathrm{OP})} \geqslant 0$ reflects the deployment and maintenance expenses; $c_i^{(\mathrm{CPU})} \geqslant 0$ captures the computational effort per operation, affecting efficiency and scalability; $c_i^{(\mathrm{MEM})} \geqslant 0$ represents the memory requirements, which are critical in resource-constrained environments; $\tau_i \geqslant 0$ denotes the latency per operation, determining execution speed and system responsiveness; and $r_i \in [0,1]$ measures the algorithm’s resilience to quantum attacks, reflecting its long-term security.

\begin{table}[ht!]
  \centering
  \caption{Two-player payoff table for players' pure strategies: row $E_i$ is a defender's encryption algorithm; column $A_j$ is an attacker's cryptanalysis method. Each cell $(i,j)$ contains\\ $\bigl((\text{success } s_{i,j},\;\text{defender's\ cost } \boldsymbol{c_i}),\;(\text{success } s_{i,j},\;\text{attacker's\ cost } k_j)\bigr)$.}
  \label{normal_form_representation}
  \resizebox{\textwidth}{!}{%
  \begin{tabular}{|c| c c c c|}
    \hline
    & $A_1$ & $A_2$ & $A_3$ & $\dots$ \\
    \hline
    $E_1$ & $\bigl((s_{1,1}, \boldsymbol{c_1}),(s_{1,1},k_1)\bigr)$ 
          & $\bigl((s_{1,2}, \boldsymbol{c_1}),(s_{1,2},k_2)\bigr)$ 
          & $\bigl((s_{1,3}, \boldsymbol{c_1}),(s_{1,3},k_3)\bigr)$ 
          & $\dots$ \\
    $E_2$ & $\bigl((s_{2,1}, \boldsymbol{c_2}),(s_{2,1},k_1)\bigr)$ 
          & $\bigl((s_{2,2}, \boldsymbol{c_2}),(s_{2,2},k_2)\bigr)$ 
          & $\bigl((s_{2,3},\boldsymbol{c_2}),(s_{2,3},k_3)\bigr)$ 
          & $\dots$ \\
    $E_3$ & $\bigl((s_{3,1}, \boldsymbol{c_3}),(s_{3,1},k_1)\bigr)$ 
          & $\bigl((s_{3,2}, \boldsymbol{c_3}),(s_{3,2},k_2)\bigr)$ 
          & $\bigl((s_{3,3}, \boldsymbol{c_3}),(s_{3,3},k_3)\bigr)$ 
          & $\dots$ \\
    $\vdots$ & $\vdots$ & $\vdots$ & $\vdots$ & $\ddots$ \\
    \hline
  \end{tabular}%
  }
\end{table}

\noindent To avoid trivial cases, the success probabilities $s_{i, j} \in (0,1)$ for $ i, j \in \{1, \dots, n\}$ reflects the inherent uncertainty in every attack method: values close to $0$ indicate a very low chance of success, while values near $1$ suggest a high likelihood of compromise, but never absolute certainty. Two qualitative observations motivate the game-theoretic analysis that follows.
\begin{itemize}
  \item If an encryption algorithm~$E_i$ has consistently low $s_{i,\cdot}$ values (i.e., low success probabilities across all applicable cryptanalysis methods), it is highly resilient even against an attacker with a sufficiently large budget. Such an algorithm is therefore an attractive option for the defender’s mixed strategy from a security perspective.
  \item If \emph{all} rows exhibit high $s_{i,\cdot}$ values  (i.e., the attacker is highly likely to breach all systems), the current algorithm portfolio~$\mathcal{A}$ is systemically weak. In this case, randomization alone provides little protection, and the defender must consider genuinely resistant alternatives.
\end{itemize}

\section{Problem formulation} \label{sec:problem_formulation}

Building on the system model of Section~\ref{sec:system_model}, we now formalize the optimization problems faced by each player. The Stackelberg structure of our problem induces a bi-level optimization issue: the attacker's subgame (the inner level) is solved first as a function of the defender's announced strategy; the defender then optimizes over the attacker's best-response map (the outer level).

\subsection{Attacker subgame} \label{subsec:attacker_subproblem}

Accounting for the observed algorithm~$E_i$, the attacker aims to optimize a \emph{weighted utility} that balances two objectives, which are maximizing the probability of breaking the encryption and minimizing the resources required to perform the attack. Assuming \emph{independent} attack attempts, the overall success probability for a chosen subset $S$ is $P_{\mathrm{succ}}(S) = 1 - \displaystyle \prod_{A_j \in S} (1 - s_{i,j})$, and the attacker must select a subset $S \subseteq \mathcal{D}(i)$ of cryptanalysis methods such that the total cost does not exceed the budget $k$, i.e., $\displaystyle\sum_{A_j \in S} k_j \leqslant k$, and satisfying,
\begin{align}
\max_{S \subseteq \mathcal{D}(i)} \quad v \cdot P_{\mathrm{succ}}(S) - \phi(S), \qquad 
\text{such that} \quad \sum_{A_j \in S} k_j \leqslant k.
\label{attacker_problem_init}
\end{align}
In problem~\eqref{attacker_problem_init}, $v > 0$ represents the value the attacker assigns to a successful compromise and $\displaystyle \phi:\ 2^{D(i)} \to \mathbb{R}, \quad S \mapsto \phi(S) = \varphi\bigg(\sum_{A_j \in S} k_j\bigg)$ 
with $\varphi:\mathbb{R}_+ \to \mathbb{R}_+$ a monotone increasing cost function that translates resource expenditure into the same utility units as $v$; that is $\forall S, S^{'} \in 2^{D(i)}$, $\displaystyle S \subseteq S^{'} \implies \varphi\bigg(\sum_{A_j \in S} k_j\bigg) \leqslant \varphi\bigg(\sum_{A_j \in S^{'}} k_j\bigg)$. 
The function $\varphi$ may be linear or convex to reflect an increasing marginal cost of acquiring resources (for instance, modeling the rapidly rising cost of additional GPUs when renting extra clusters). 
\begin{remark}
 We assume that the function $ - \phi$ is submodular so that the attacker’s objective preserves the desirable diminishing-returns property. Some examples include the linear form $\phi(S) = \varphi\Bigl(\displaystyle \sum_{A_j\in S} k_j\Bigr) = \lambda \sum_{A_j\in S} k_j$ and convex, increasing forms such as $\phi(S) = \varphi\Bigl(\displaystyle \sum_{A_j\in S} k_j\Bigr) = \lambda \Bigl(\displaystyle \sum_{A_j\in S} k_j\Bigr) + \alpha \Bigl(\displaystyle \sum_{A_j\in S} k_j\Bigr)^2$, $\lambda, \alpha \geqslant 0$. 
\end{remark} \nobreak
\noindent The problem~\eqref{attacker_problem_init} can therefore be written as the following optimization problem:
\begin{align*}
\text{maximize} \quad  & F(S) = v \cdot \Bigl(1 - \prod_{A_j \in S}(1 - s_{i,j})\Bigr) - \phi(S) \\
\text{subject to} & \sum_{A_j\in S} k_j \leqslant k,
\end{align*}
    
\noindent The attacker's optimization problem for a realized algorithm~$E_i$ is therefore expressed as,
\begin{align}
  S^\star(i) = \arg\max_{S \subseteq \mathcal{D}(i)}\;
    v \cdot \Bigl(1 - \prod_{A_j \in S}(1 - s_{i,j})\Bigr) - \phi(S)
  \quad \text{such that} \quad \sum_{A_j\in S} k_j \leqslant k
  \label{attacker_problem}
\end{align}\nobreak

\begin{remark}
 The objective in~\eqref{attacker_problem} mixes a \emph{value} term $v \cdot P_{\mathrm{succ}}(S)$ with a \emph{resource-cost} term $\displaystyle\sum_{A_j\in S} k_j$. An alternative to the normalization approach adopted above is to express the problem as $\displaystyle \max_{S\subseteq\mathcal D(i)}\; P_{\mathrm{succ}}(S) - \lambda \sum_{A_j\in S} k_j$, where $\lambda$ converts resource cost into probability units.
\end{remark}

\noindent The attacker’s optimization problem, as formulated above, exhibits several structural properties that are crucial for its analysis and solution.
\noindent \paragraph{Submodularity property} 
	
\noindent Consider the function $g : 2^{D(i)} \to \mathbb{R}$ such that $g(S) = P_{\mathrm{succ}}(S)$, then 
\begin{align*}
		&g(A \cup \{A_j\}) - g(A) = 1 - \prod_{A_r \in A \cup \{A_j\}} (1 - s_{i, r}) - \left(1 - \prod_{A_r \in A} (1 - s_{i, r}) \right) = s_{i, j} \cdot \prod_{A_r \in A} (1 - s_{i, r}),\\
		&g(B \cup \{A_j\}) - g(B) = s_{i, j} \cdot \prod_{A_r \in B} (1 - s_{i, r});\ \text{ since $0 < 1 - s_{i, j} < 1$ and $A \subseteq B$,}\\
        &\prod_{A_r \in A} (1 - s_{i, r}) \geqslant \prod_{A_r \in B} (1 - s_{i, r}) \quad \Rightarrow \quad g(A \cup \{A_j\}) - g(A) \geqslant g(B \cup \{A_j\}) - g(B).
\end{align*}\nobreak
\noindent Hence $g$ is submodular. Adding a submodular function $ -  \phi$ preserves submodularity. Therefore the function
$F : 2^{D(i)} \to \mathbb{R}, \quad F(S) = v \cdot g(S) -\phi(S), \ \forall S \subseteq D(i)$ is submodular.

\paragraph{Non-monotonicity} For any $j \notin S$,
\begin{align*}
     F(S \cup \{j\}) - F(S) & = v\cdot (g(S\cup\{A_j\}) - g(S)) - \phi(S \cup\{A_j\}) + \phi(S) \quad \text{i.e., }\\
     F(S \cup \{j\}) - F(S) \leqslant 0 & \iff \phi(S \cup\{A_j\}) - \phi(S) \geqslant v\cdot (g(S \cup \{A_j\}) - g(S))
\end{align*}\nobreak
\noindent Hence $F$ is
\emph{non-monotone}.  Together with the budget constraint
$\displaystyle \sum_{A_j\in S} k_j \leqslant k$, problem~\eqref{attacker_problem} is a
\emph{non-monotone submodular maximization problem under a knapsack constraint}.

\subsection{Defender Optimization Problem}
\label{subsec:defender_problem}

Let $P_{\mathrm{succ}}^\star(i)$ denote the optimal attacker success probability
against algorithm~$E_i$, achieved by the optimal best-response set $S^\star(i)$ obtained by solving~\eqref{attacker_problem}.  The defender evaluates each $E_i$ using a \emph{per-algorithm utility} value,
\begin{align}
  \ell_i = L_i \cdot \bigl(1 - P_{\mathrm{succ}}^\star(i)\bigr)
          - \gamma_{\mathrm{OP}}\, c_i^{(\mathrm{OP})}
          - \gamma_{\mathrm{CPU}}\, c_i^{(\mathrm{CPU})}
          - \gamma_{\mathrm{MEM}}\, c_i^{(\mathrm{MEM})}
          - \gamma_{\tau}\,\tau_i
          + \gamma_r\, r_i,
  \label{per_algo_utility}
\end{align}\nobreak
where the parameters are defined in Table~\ref{tab:algo_params} while the defender specifies weights translate each parameter into utility units, as explained in Table \ref{weights_optimization} below.

\begin{longtable}{p{2.5cm}p{6.5cm}p{4.25cm}}
\caption{Algorithm-specific parameters in the defender's optimization problem. \label{tab:algo_params}} \\
\toprule
\textbf{Parameter} & \textbf{Meaning} & \textbf{Obtainable from}  \\
\midrule
\endfirsthead

\bottomrule
\endlastfoot

$L_i > 0$ & \emph{Value protected}: asset value remaining secure if algorithm~$E_i$ successfully resists the attack. Higher $L_i$ indicates more critical data/systems protected by algorithm~$E_i$ & business impact analysis (BIA) \\[2pt]

$P_{\mathrm{succ}}^\star(i) \in [0,1]$ & \emph{Attacker's optimal success probability}: maximum probability that a rational, resource-constrained attacker compromises algorithm~$E_i$ & the attacker subgame optimization (i.e., the value $P_{\text{succ}}(S^{\star}(i))$) \\[2pt]

$c_i^{(\mathrm{OP})} \geqslant 0$ & \emph{Operational/setup cost}: fixed deployment costs associated with algorithm~$E_i$ & costs for purchasing hardware + running costs to maintain it; from business continuity management \\[2pt]

$c_i^{(\mathrm{CPU})} \geqslant 0$ & \emph{CPU cost per operation}: computational cycles required for one encryption algorithm & binary code for algorithms: depending on programming language and compiler optimizations \\[2pt]

$c_i^{(\mathrm{MEM})} \geqslant 0$ & \emph{Memory footprint}: RAM for keys, intermediate values, and algorithm state & space complexity of the algorithm (as given in the literature) \\[2pt]

$\tau_i \geqslant 0$ & \emph{Latency per operation}: wall-clock time including all overheads & studies conducted in the literature \\[2pt]

$r_i \in [0,1]$ & \emph{Quantum-resilience score}: assessment of algorithm's resistance to quantum attacks; $r_i = 1$ means quantum-safe (e.g.\ AES-256, Kyber-1024). $r_i \in (0,1)$: partial resilience or uncertain quantum security & expert elicitation (statistics) \\
\end{longtable}

\begin{longtable}{p{2cm}p{11.75cm}}
\caption{Defender's utility weights for different resource and strategic dimensions. \label{weights_optimization}} \\
\toprule
\textbf{Weight} & \textbf{Interpretation} \\
\midrule
\endfirsthead

\bottomrule
\endlastfoot

$\gamma_{\mathrm{OP}} \geqslant 0$ & \emph{Operational-cost weight}: converts monetary setup/maintenance costs into defender utility. \\[3pt]

$\gamma_{\mathrm{CPU}} \geqslant 0$ & \emph{CPU-cost weight}: penalty per GHz$\cdot$ms of computation; reflects server costs and energy. \\[3pt]

$\gamma_{\mathrm{MEM}} \geqslant 0$ & \emph{Memory-cost weight}: penalty per MB of RAM; critical in memory-constrained IoT devices. \\[3pt]

$\gamma_{\tau} \geqslant 0$ & \emph{Latency-penalty weight}: penalty per millisecond of delay. \\[3pt]

$\gamma_r \geqslant 0$ & \emph{Quantum-resilience reward}: bonus per unit of quantum resistance; encodes long-term security vision. \\
\end{longtable}

\noindent The defender then chose a mixed strategy $\boldsymbol{p} = (p_1, \ldots, p_n)$ to solve
\begin{align}  \label{defender_optimization_problem}
    \boldsymbol{p}^\star \;=\; &\arg\max_{\boldsymbol{p} \in \Delta_n}\ \sum_{i=1}^n p_i\,\ell_i
        \qquad (\ell_i \text{ given by } \eqref{per_algo_utility}) \notag \\[5pt]
    \text{subject to:}\quad
    &\sum_{i=1}^n p_i = 1
        (\text{simplex})\notag \\ \;
    &\sum_{i=1}^n p_i\,c_i^{(\mathrm{OP})} \leqslant
        C_{\mathrm{OP}} \ (\text{Operational cost}); \quad
    \sum_{i=1}^n p_i\,c_i^{(\mathrm{CPU})} \leqslant C_{\mathrm{CPU}} \
        (\text{CPU budget}) \notag \\
    &\sum_{i=1}^n p_i\,c_i^{(\mathrm{MEM})} \leqslant C_{\mathrm{MEM}} \
        (\text{memory budget}); \quad
    \sum_{i=1}^n p_i\,\tau_i \leqslant T_{\max} \ 
        (\text{latency constraint}) \notag \\
    &\sum_{i=1}^n p_i\,r_i \geqslant R_{\min} \
        (\text{quantum resilience}); \quad \ p_i \geqslant 0, \ i=1,\ldots,n\
        (\text{non-negativity}) \notag \\
    &\sum_{i \in \mathcal{F}_j} p_i \leqslant \alpha_j,\quad j=1,\ldots,J \quad
        (\text{enforcing diversification})
\end{align} \nobreak
\noindent The global constraints and their operational motivations are summarized in the
Table~\ref{constraints_optimization} below.

\begin{longtable}{p{3.5cm}p{9.5cm}}
\caption{Global constraints and their strategic motivations.} \label{constraints_optimization} \\
\toprule
\textbf{Constraint} & \textbf{Meaning} \\
\midrule
\endfirsthead

\toprule
\textbf{Constraint} & \textbf{Meaning} \\
\midrule
\endhead

\bottomrule
\endlastfoot
$C_{\mathrm{OP}} > 0$ & Maximum operational-cost available to the defender. \\[3pt]
$C_{\mathrm{CPU}} > 0$ & Maximum expected CPU resources per operation. \\[3pt]

$C_{\mathrm{MEM}} > 0$ & Maximum RAM available for cryptographic operations. \\[3pt]

$T_{\max} > 0$ & Upper bound on expected operation latency. \\[3pt]

$R_{\min} \in [0,1]$ & Required average quantum-resistance score. \\[3pt]

$\mathcal{F}_j \subseteq \{1,\ldots,J\}$ & Subset of algorithms sharing a common cryptographic foundation (e.g., $\mathcal{F}_1 = \{\text{RSA2048, RSA4096}\}$ (RSA family); $\mathcal{F}_2 = \{\text{Kyber, Dilithium}\}$ (lattice family)). \\[3pt]

$\alpha_j \in (0,1)$ & Maximum total probability allocated to family~$j$; ensures diversification across cryptographic families (e.g., $\alpha_j = 0.5$ prevents over 50\% concentration in any family). \\
\end{longtable}

\noindent We note that the problem~\eqref{defender_optimization_problem} is a LP with $n$ decision variables (probabilities $p_1, \ldots, p_n$), one equality constraint (simplex), five resource-inequality constraints (operational cost, CPU, memory, latency, quantum resilience), $J$ family-concentration inequalities, and $n$ non-negativity constraints.

\begin{remark} 
We adopt a constrained optimization framework with multiple constraints, since relying on a single aggregate cost is often too restrictive. Real cryptographic systems involve heterogeneous resources such as CPU, memory, bandwidth, and energy, which a single scalar measure cannot capture. Furthermore, the need for quantum resistance, driven by ``harvest-now-decrypt-later” threats and NIST recommendations, makes it necessary to incorporate post-quantum algorithms. NIST guidelines emphasize the importance of cryptographic agility, i.e., the ability to transition to new cryptographic algorithms when existing ones become vulnerable \cite{barker2018nist}. Finally, a model with only one constraint typically yields solutions with support size $|\mathrm{supp}(p^\star)|\leqslant 2$, whereas in practice deploying only one or two algorithms increases the risk of correlated failures; the multi-constraint formulation naturally produces more diverse portfolios.
\end{remark}

\section{Game Solution} \label{sec:game_solution}

This section presents the solution to the Stackelberg game defined in Section~\ref{sec:problem_formulation}. We first establish that an equilibrium always exists and characterize the support of the optimal defender strategy.  We subsequently describe the algorithm used to solve the attacker's subgame.
\medskip

\subsection{Stackelberg equilibrium and support size of the defender}

\noindent \emph{Existence of Stackelberg equilibrium.} For a given encryption algorithm $E_i$, the attacker chooses a set of attack methods $S$ that maximizes their utility while staying within the budget. The \emph{attacker’s best response} is defined as $S^{*}(i) \in \displaystyle\argmax_{S \subseteq \mathcal{D}(i)} F(S)$, and knowing that, the defender anticipates the attacker’s choice and selects a mixed strategy $\boldsymbol{q} \in \Delta_n$ over the encryption algorithms to maximize their expected utility. Therefore, a pair of strategies $(\boldsymbol{q}^{*}, S^{*}(\cdot))$ is a \emph{Stackelberg equilibrium} of the game if $\quad \displaystyle\boldsymbol{q}^{*} \in  \argmax_{\boldsymbol{q}\in\Delta_n}  \sum_{i=1}^{n} q_i \, \ell_i, \text{ where } \forall i \in \{1, \dots, n\},\ S^{*}(i) \in \argmax_{S \subseteq \mathcal{D}(i)}F(S)$; neither player can improve its payoff by unilaterally deviating from these strategies.

\begin{theorem}\label{Existence_Stackelberg_equilibrium}
  The cryptographic hybridization game always admits a Stackelberg equilibrium
  $\left(\boldsymbol{p}^\star,\{S^\star(i)\}_{i=1}^{n}\right)$.
\end{theorem}

\begin{proof}
  The following conditions are verified.
  \begin{enumerate}
    \item \emph{Defender's strategy space:} $\Delta_n$ is non-empty, compact, and
          convex.
    \item \emph{Attacker's best response:} for each observed~$i$, the attacker maximizes a continuous function over the finite set $2^{\mathcal{D}(i)}$. A maximum exists; the optimum $S^\star(i)$ may not be unique, but $P_{\mathrm{succ}}^\star(i)$ is well-defined.
    \item \emph{Defender's objective:} given attacker's best responses, $\ell_i(P_{\mathrm{succ}}^\star(i))$ is fixed for each~$i$.  The objective $\displaystyle \sum_{i = 1}^{n} p_i\ell_i$ is linear (hence continuous and convex) in~$\boldsymbol{p}$.
    \item \emph{Constraints:} all resource and family constraints are linear,
          preserving convexity.
  \end{enumerate}
\noindent By the Weierstrass extreme-value theorem~\cite{martinez2014weierstrass}, a continuous function on a non-empty compact set attains its maximum. Hence, $\boldsymbol{p}^\star$ exists, and, together with the attacker's best responses, it constitutes a Stackelberg equilibrium.
\end{proof}

\noindent \emph{Support size of the optimal defender strategy.} The support size bound of the defender's optimal strategy is specified by the following Theorem~\ref{Support_size_bound}.

\begin{theorem}\label{Support_size_bound}
  \label{thm:extended_support}
  Let $\boldsymbol{p}^\star$ be an optimal solution to the problem~\eqref{defender_optimization_problem}, and let $u$ denote the number of constraints (excluding non-negativity) that are binding (active with equality) at~$\boldsymbol{p}^\star$, then $|\mathrm{supp}(\boldsymbol{p}^\star)| \;\leqslant\; u,
    \ \text{ with } \mathrm{supp}(\boldsymbol{p}^\star) = \left\{i \in \{1, \dots, n\} : p_i^\star > 0\right\}$.
\end{theorem}

\begin{proof}
  We use the theory of Basic Feasible Solutions (BFS) in linear programming.
  Introducing slack variables $s_0, s_1,\ldots,s_{4+J}$ for the inequality constraints,
  problem~\eqref{defender_optimization_problem} is equivalent to the standard-form LP
  \begin{align*}
     \text{maximize} & \quad \sum_{i=1}^n p_i\,\ell_i \\
     \text{subject to} & \quad
       \sum_{i=1}^n p_i = 1, \quad \sum_{i=1}^n p_i\,c_i^{(\mathrm{OP})} + s_{0} = C_{\mathrm{OP}} \quad \sum_{i=1}^n p_i\,c_i^{(\mathrm{CPU})} + s_1 = C_{\mathrm{CPU}},\\
       &\sum_{i=1}^n p_i\,c_i^{(\mathrm{MEM})} + s_2 = C_{\mathrm{MEM}}, \quad
       \sum_{i=1}^n p_i\,\tau_i + s_3 = T_{\max},\quad  \sum_{i=1}^n p_i\,r_i - s_4 = R_{\min},\\
      & \sum_{i \in \mathcal{F}_j} p_i + s_{4+j} = \alpha_j,\quad j=1,\ldots,J,\quad  p_i,\, s_\ell \geqslant 0 \quad \forall\, i,\ell.
  \end{align*}
This system has $n + (5+J)$ variables and $1 + (5+J)$ equality constraints. A BFS is a feasible point where exactly $m = 1 + (5 + J)$ variables are \emph{basic} (potentially nonzero) and the remaining $(n + 5 + J) - m = n - 1$ variables are \emph{non-basic} ($=0$).
By the fundamental theorem of LP, if an optimal solution exists, there is an optimal BFS. At an optimal BFS $\boldsymbol{p}^\star$, a constraint is binding (active) if its slack variable $s_\ell = 0$, and the number of binding constraints equals the number of basic variables among $\{p_1, \ldots, p_n\}$ plus the number of basic slack variables that are zero.
More precisely, let $u$ the number of constraints (among the $1 + 5 + J$ equalities) that are binding at $\boldsymbol{p}^\star$, then,
\begin{itemize}
    \item $u$ of the slack variables is zero (non-basic or basic at zero),
    \item The remaining $(1 + 5 + J) - u$ slack variables are positive (basic and nonzero).
    \item Since we have $1 + 5 + J$ equality constraints (a basis) and $(1 + 5 + J) - u$ slacks are basic and positive, the remaining $u$ basic variables must come from $\{p_1^\star, \ldots, p_n^\star\}$.
\end{itemize}\noindent
Therefore, at most $u$ probabilities are strictly positive, i.e., $|\text{supp}(p^\star)| = |\{i : p_i^\star > 0\}| \leqslant u.$
\end{proof}

\subsection{Solving the Attacker subgame}
\label{subsec:attacker_algorithm}

As established in Section~\ref{sec:problem_formulation}, problem~\eqref{attacker_problem} is a non-monotone submodular maximization problem under a knapsack constraint.  Table~\ref{literature_knapsack} surveys several approximation algorithms available in the literature for this class of problem.

\begin{table}[ht!]
  \centering
\caption{Several algorithms for non-monotone submodular maximization
  under a knapsack constraint; these results hold for the general problem class,
  of which our attacker subgame~\eqref{attacker_problem} is a special instance.}
  \label{literature_knapsack}
  \begin{tabular}{l c c c c}
    \toprule
    \textbf{Reference} & \textbf{Objective} & \textbf{Constraint}
        & \textbf{Approximation\ Ratio} & \textbf{Queries} \\
    \midrule
    Ene et al.~\cite{ene2017nearly} & General & Knapsack & $e$
        &  $\tilde{\mathcal{O}}(n^2)$ \\
    SampleGreedy~\cite{amanatidis2020_fast} & General & Knapsack & $5.83$
        &  $\tilde{\mathcal{O}}(n\log n)$ \\
    ParKnapsack~\cite{amanatidis2021} & General & Knapsack & $9.465$
        &  $\tilde{\mathcal{O}}(n)$ \\
    Cui et al.~\cite{cui}            & General & Knapsack & $8$
        & $\tilde{\mathcal{O}}(n)$ \\
    \bottomrule
  \end{tabular}
\end{table}

\noindent Although these algorithms represent significant theoretical advances, they could be poorly suited to our context, in terms of the precision of the solution obtained, for several reasons.

\begin{itemize}
    \item \emph{Instance size is small:} Unlike large combinatorial problems (e.g., 62,000 videos in Amanatidis et al.~\cite{amanatidis2020_fast}), each encryption algorithm~$E_i$ admits only a limited and well-defined set of cryptanalytic techniques; for instance, symmetric ciphers such as AES, public-key schemes like RSA, and post-quantum lattice-based constructions are subject to a finite catalog of attack classes (e.g., differential, algebraic, lattice reduction, etc.), as documented in \cite{boneh1999twenty,daemen2002design}. Consequently, for each $E_i$, $|\mathcal{D}(i)|$ remains relatively small in practice, making exact methods such as dynamic programming computationally feasible, fast, and memory-efficient.
    
    \item \emph{Exploitable structure:} The attacker’s success probability has a closed form $P_{\mathrm{succ}}(S) = 1 - \displaystyle \prod_{A_{j}\in S}(1 - s_{i,j})$, allowing incremental updates in $O(1)$ per-added method. This structure reduces complexity to $O(|\mathcal{D}(i)|\cdot k)$, unlike generic black-box approximation algorithms.

    \item \emph{Security implications of approximation:} Approximation factors can underestimate adversary strength. For example, a $5.83$-approximation could yield only $\approx 17\%$ of the optimal attack success, leading to systematically underestimated threats and potentially inadequate defenses (thereby creating security vulnerabilities). Real runs, as illustrated in our comparative example below, confirm this risk. In this particular run of \textsc{SampleGreedy}, the achieved $71.8\%$~of the optimal value represents a non-negligible underestimation of the true adversarial threat.
\end{itemize}

\noindent \emph{Illustrative example with four algorithms.} \label{Inneficiency_example}
To illustrate the performance gap between \textsc{SampleGreedy} and the exact DP method in our context with relatively small cryptographic instances, we present the following comparative example. Consider an attacker targeting $AES128$ ($E_{i} =$ $AES128$) with four attack methods and the parameters in Table~\ref{example} (selected only for illustration reasons and do not stem from real data):
\begin{table}[!ht]
\centering
\caption{\textbf{Instance parameters:} $v = 1200$, $k = 500$, $\varphi(x) = x$ (linear), $q = 0.414$ (optimal probability from~\cite{amanatidis2020_fast})}
\label{example}
\begin{tabular}{c c c l}
\toprule
\textbf{Method} & \textbf{Success $s_{i,j}$} & \textbf{Cost $k_j$} & \textbf{Description} \\
\midrule
$A_1$ & 0.20 & 100 & Brute force variant \\
$A_2$ & 0.35 & 200 & Linear cryptanalysis \\
$A_3$ & 0.42 & 280 & Differential cryptanalysis \\
$A_4$ & 0.25 & 120 & Side-channel attack \\
\bottomrule
\end{tabular}
\end{table}

\noindent \textsc{SampleGreedy} first identifies the best single method ($A_3$, utility 224.0), 
then runs the greedy construction: $A_4$ (best density $(v \cdot s_{i, 4} - k_4) / k_4 = 1.50$) is rejected at iteration~1; 
$A_1$ (density 1.40) is accepted at iteration~2; $A_2$ is then rejected; all remaining 
candidates are exhausted. The algorithm returns $S^* = \{A_1\}$ with 
$F(S^*) = 1200 \times 0.20 - 100 = 140$. Since the best single method $A_3$ dominates 
(utility 224.0), \textsc{SampleGreedy} finally outputs $\{A_3\}$ with utility $\emph{224.0}$.\\ \\
\noindent \emph{Application of DP:} The DP algorithm explores all feasible combinations systematically to find the optimal solution $\mathrm{DP}[i, k] =\max\left\{F(S)~\mid~\substack{S \subseteq \mathcal{D}_i},\displaystyle \sum_{A_j \in S} k_j \leqslant k \right\}$ over the set $\{A_1, A_2, A_3, A_4\}$. Its application of DP here returns the optimal solution $S^* = \{A_1, A_2, A_4\}$ with 
$P_{\mathrm{succ}} = 1-(0.80 \cdot 0.65 \cdot 0.75) = 0.61$ and $F(S^*) = 1200 \times 0.61 - 420 = \emph{312.0}$, 
a \emph{28.2\% improvement} over \textsc{SampleGreedy}. Full execution details are reported in Appendix~\ref{app:samplegreedy}.

\medskip

In terms of security implications, the optimal attack achieves a success probability of $61\%$ whereas \textsc{SampleGreedy} algorithm approximates $42\%$. Consequently, using \textsc{SampleGreedy} to estimate adversarial capability underestimates the threat by $19\%$, potentially resulting in inadequate security margins; the DP solution provides the defender with the true worst-case risk. We therefore propose to solve our attacker subgame with a \emph{hybrid algorithm} that chooses between DP and greedy based on problem size: If $|\mathcal{D}(i)|\leqslant \tau$ (threshold for an encryption algorithm~$ E_i$), then DP is applied, otherwise we employ the \textsc{SampleGreedy} algorithm. The proposed DP for the attacker subgame has a \emph{pseudo-polynomial} complexity $O(n \times k)$, where $n = |\mathcal{D}(i)|$ and and $k$ is the attacker budget, as established in \cite{axiotis2018capacitated}, and is tractable when the capacity parameter remains relatively small (typically $k\leqslant 10^5$), while larger values quickly render the approach infeasible due to memory and time constraints. That is, the effective DP table size $n \times k$ must remain within moderate bounds (on the order of $10^4$ – $10^6$) for practical use \footnote{https://www.w3tutorials.net/blog/why-is-the-knapsack-problem-pseudo-polynomial/}. Since the DP runtime grows linearly with the numerical capacity, solving larger instances becomes prohibitive, especially when the attacker problem must be solved repeatedly. Therefore, in our experiments, we restrict DP usage to instances that can be solved within a small, constant runtime, specifically $0.2$ seconds to find the switching threshold $\tau$, as shown in Figure~\ref{threshold}. This runtime selection is an illustrative heuristic rather than a general rule; alternative approaches or time limits could be adopted. For our experiments, we conducted the threshold (timing) analysis on a machine with the following capabilities:
\begin{itemize}
    \item CPU: Intel Core i5-1135G7, 4 physical cores / 8 threads, 2.4 GHz
    \item RAM: 16 GB
    \item OS: Windows 10, 64-bit
\end{itemize}
We note that the computation time of the algorithm depends on numerous parameters, notably the method costs $k_j$ and the success probabilities $s_{., j}$. For the fixed parameter sets used here, we obtained a threshold $\tau = 310$ as shown in Figure~\ref{threshold}. This threshold is not universal and may change depending on the parameter values or computational hardware; our goal is only to illustrate a procedure for determining $\tau$.

\begin{figure}[h!]
\vspace{-0.25cm}
\centering
\includegraphics[width=0.75\linewidth]{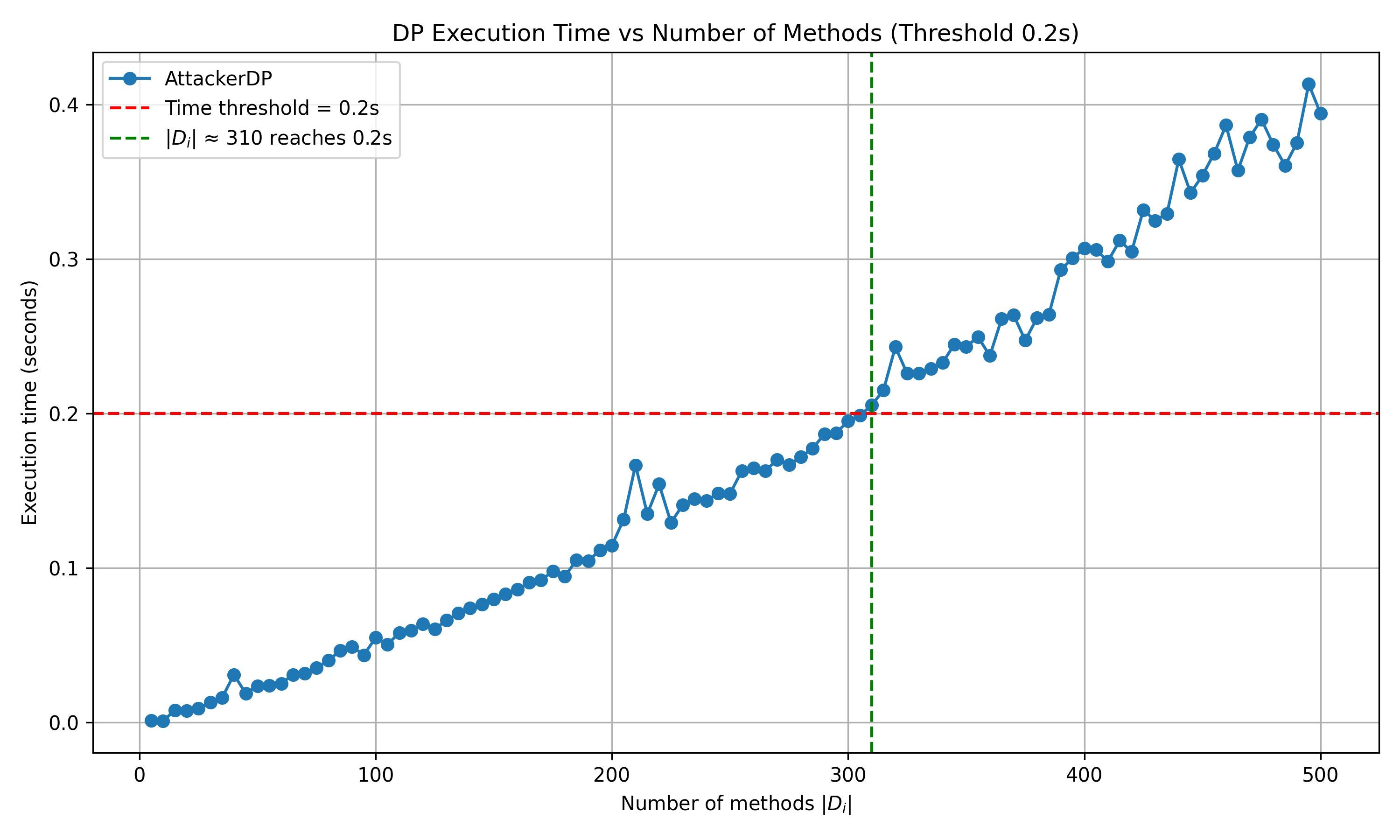}
\caption{\emph{Threshold approximation for the AttackerDP algorithm:}  
Using a DP execution time limit of $0.2$ seconds, the threshold number of methods is approximately $310$. This threshold is specific to this simulation and the chosen parameters (number of methods, budget, success probabilities, and costs) and may vary if the parameters or hardware change. Our parameters are as follows:\\[2mm]
\textbullet\ DP benchmark: maximum number of methods tested = $500$, $k$ = $500$, value $v = 1000$\\
\textbullet\ Success probability range: $s_{., j} \in [0.05, 0.85]$,\ \ Method cost range: $k_j \in [40, 200]$,\ \ Cost function: $\varphi(x) = x$}
\label{threshold}
\end{figure}

\noindent The pseudo-code for our DP is presented in Algorithm~\ref{algorithm_attacker_dp}, while the hybrid algorithm is outlined in Algorithm~\ref{algorithm_hybrid_attacker}.

\begin{breakablealgorithm}
\caption{\textsc{AttackerDP}($\mathcal{D}(i), \{s_{., j}\}, \{k_j\}, k, v, \phi$)}
\label{algorithm_attacker_dp}

\begin{algorithmic}[1]
    \Input Methods $\mathcal{D}(i) = \{p_1,\ldots,p_n\}$, success probabilities $\{s_{., j}\}$, costs $\{k_j\}$, budget $k$, value $v$, cost function $\phi$
    \Output Optimal subset $S^*$

    \State $n \leftarrow |\mathcal{D}(i)|$

    \State \Comment{\textcolor{blue}{Create DP table: $dp[i][c] = (max\_utility, failure\_prob, selected\_set)$}}
    \State Create table $dp[0..n][0..k]$

    \State \Comment{\textcolor{blue}{Base case: no methods, no cost}}
    \For{$c = 0$ to $k$}
        \State $dp[0][c] \leftarrow (0, 1.0, \emptyset)$
    \EndFor

    \State \Comment\textcolor{blue}{{Fill table for each method}}
    \For{$i = 1$ to $n$}
        \For{$c = 0$ to $k$}
            \State \Comment\textcolor{blue}{{Option 1: Don't take method $p_i$}}
            \State $(F_{\text{skip}}, FP_{\text{skip}}, S_{\text{skip}}) \leftarrow dp[i-1][c]$

            \State \Comment{\textcolor{blue}{Option 2: Take method $p_i$ (if cost allows)}}
            \If{$k_i \leqslant c$}
                \State $(F_{\text{prev}}, FP_{\text{prev}}, S_{\text{prev}}) \leftarrow dp[i-1][c - k_i]$
                \State $FP_{\text{new}} \leftarrow FP_{\text{prev}} \cdot (1 - s_{., i})$
                \State $P_{\text{succ}} \leftarrow 1 - FP_{\text{new}}$
                \State $F_{\text{take}} \leftarrow v \cdot P_{\text{succ}} - \phi(c)$
                \State $S_{\text{new}} \leftarrow S_{\text{prev}} \cup \{p_i\}$

                \State \Comment{\textcolor{blue}{Choose better option}}
                \If{$F_{\text{take}} > F_{\text{skip}}$}
                    \State $dp[i][c] \leftarrow (F_{\text{take}}, FP_{\text{new}}, S_{\text{new}})$
                \Else
                    \State $dp[i][c] \leftarrow dp[i-1][c]$
                \EndIf
            \Else
                \State $dp[i][c] \leftarrow dp[i-1][c]$
            \EndIf
        \EndFor
    \EndFor

    \State $(F^*, \_, S^*) \leftarrow \displaystyle\max_{c \leqslant k} dp[n][c]$ by utility  \Comment{\textcolor{blue}{Find best solution across all valid costs}}
    \State \Return $S^*$
\end{algorithmic}
\end{breakablealgorithm}

\begin{breakablealgorithm}
  \caption{\textsc{HybridAttacker}($\mathcal{D}_i,\, s,\, k_{\text{costs}},\, k_{\text{budget}},\, v,\, \phi$, $\tau$)}
    \label{algorithm_hybrid_attacker}
  \begin{algorithmic}[1]
    \Input Set of cryptanalysis methods $\mathcal{D}_i$ for algorithm $E_i$; success probabilities $s_{., j}$ and resource costs $k_j$ for each method $A_j$; total budget $k$; value of successful breach $v$; cost function $\phi$; threshold $\tau$ for switching between \textsc{AttackerDP} and \textsc{SampleGreedy}
    \Output Selected subset $S \subseteq \mathcal{D}_i$

    \If{$k_{\text{budget}} \cdot |\mathcal{D}_i| \leqslant 10^{5}$\ and \ $\tau \leqslant \tau'$}
      \State \Return $\textsc{AttackerDP}(\mathcal{D}_i, s, k_{\text{costs}}, k_{\text{budget}}, v, \phi)$
    \Else
      \State \Return $\textsc{SampleGreedy}(\mathcal{D}_i, s, k_{\text{costs}}, k_{\text{budget}}, v, \phi)$
    \EndIf
  \end{algorithmic}

\end{breakablealgorithm}

In the previous analysis, the defender is assumed to possess precise knowledge of the attacker’s budget, which determines the attacker’s operational capabilities. Such an assumption, however, may be unrealistic in practical security settings where adversarial resources are inherently uncertain. We therefore introduce an alternative decision framework in which the defender must act under uncertainty about the attacker’s budget. The proposed approach relies on a worst-case regret minimization principle: rather than optimizing performance for a specific assumed budget, the defender selects a strategy that minimizes the maximum performance loss (regret) relative to the optimal strategy that would have been chosen if the true attacker budget were known a priori. In other words, to limit the penalty induced by incorrect assumptions about adversarial capabilities, the defender seeks a robust allocation that performs as close as possible to the budget-aware optimal solution.

\section{Robust optimization and regret minimization: Uncertain attacker budget}
\label{sec:worst_case_regret_minimization}

In practice, the defender may not know the attacker's budget $k$ precisely.  Since the utility associated with each algorithm~$E_i$ depends on $k$ through the attacker's success probability, we write
\begin{align} \label{eq:ellK}
      \ell_i(k) = L_i\cdot\bigl(1 - P_{\mathrm{succ}}^\star(i,k)\bigr)
              - \gamma_{\mathrm{OP}}\,c_i^{(\mathrm{OP})}
              - \gamma_{\mathrm{CPU}}\,c_i^{(\mathrm{CPU})}
              - \gamma_{\mathrm{MEM}}\,c_i^{(\mathrm{MEM})}
              - \gamma_{\tau}\,\tau_i
              + \gamma_r\,r_i,
\end{align}
\noindent where \(P_{\mathrm{succ}}^\star(i, k)\) is the attacker's probability of success against the encryption algorithm~\(E_i\) under a budget constraint \(k\). When the defender does not know $k$, the optimization problem must account for this uncertainty. Denoting by $K$ the set of possible attacker budgets, the defender must then select a mixed strategy $\boldsymbol{p}$ that performs well across all possible values of $k \in K$ and, we thus extend the defender formulation \eqref{defender_optimization_problem} using two approaches from robust optimization that seem relevant in this context: the \emph{absolute worst-case optimization} and the \emph{minimax regret}. These approaches differ in the degree of conservativeness the defender wishes to adopt.

\paragraph{(A) Absolute worst-case objective (maximin)}   The defender chooses \(\boldsymbol{p}\in\Delta_n\) to minimize the worst-case expected loss that could arise as $k$ ranges over \(K\). This is conservative, since it optimizes performance against the most damaging attacker budget in the set. The defender then solves, \quad
$\boldsymbol{p}^{\mathrm{wc}}
= \displaystyle \arg\max_{\boldsymbol{p}\in\Delta_n}  \min_{k\in K}  \sum_{i=1}^n p_i \ell_i(k)$, \ ensuring that the defender’s performance is as good as possible in the worst possible budget scenario
\medskip

\noindent \textbf{Special case:} If in particular, the attacker is \emph{unconstrained} with no effective budget limitation (i.e. can run all available methods), then for each encryption algorithm~\(E_i\) the success probability is $P_{\mathrm{succ}}^\star(i,\infty) \;=\; 1 - \displaystyle \prod_{j\in\mathcal{D}(i)} \bigl(1 - s_{j, i}\bigr),$ leading to the following LP~\ref{eq:wE_ill},
\begin{equation}\label{eq:wE_ill}
\begin{aligned}
& \max_{\boldsymbol{p}\in\Delta_n}\quad \sum_{i=1}^n p_i \ell_i(\infty) \quad \text{such that} \\[4pt]
& \text{$\Delta_n$, operational and CPU budgets, latency, quantum-resilience, family limits}\ \text{as in } \eqref{defender_optimization_problem}.
\end{aligned}
\end{equation}
where $\ell_i(\infty) \;=\; \ell_i\bigl(k = \infty\bigr)$ by leveraging \eqref{eq:ellK}.
\noindent This is an ordinary LP (the same constraint matrix as \eqref{defender_optimization_problem}) and is solved directly by any LP solver.
\begin{remark}
Optimizing against an unconstrained attacker is equivalent to preparing for the strongest possible adversary, including one that could try every available cryptanalytic method. While this provides a very conservative safeguard (when the defender must guarantee performance against the strongest possible attacker), it may be more than necessary when the attacker’s budget is believed to be finite but uncertain within a range \(K\). In such situations, focusing solely on the worst-case scenario can result in overly cautious strategies designed for the largest conceivable budget. A minimax regret formulation is then preferable because it selects a strategy that remains close to optimal across all plausible budget levels.
\end{remark}

\paragraph{(B) Minimax regret objective} Believing that the previous formulation may lead to overly conservative policies because it optimizes only for the worst-case scenario, in practice, if the attacker’s budget turns out to be smaller than assumed, the defender may unnecessarily sacrifice performance (for instance, by selecting overly costly or slow algorithms). An alternative is to minimize standard worst-case regret: \emph{rather than protecting against the worst possible outcome, the defender compares the performance of a chosen strategy with the strategy that would have been optimal if the true attacker budget had been known in advance}. Let \(V^\star(k)\) denote the defender's optimal value when the attacker's budget \(k\) is known in advance, that is $V^\star(k) = \displaystyle \max_{\boldsymbol{q}\in\Delta_n} \sum_{i=1}^n q_i \, \ell_i(k).$
The defender then solves the minimax (absolute) regret problem~\ref{eq:robust_regret_cont}:
\begin{align}\label{eq:robust_regret_cont}
    \boldsymbol{p^{\mathrm{mmr}}} \;=\; \arg\min_{\boldsymbol{p}\in\Delta_n} \; \max_{k\in K} \; \Big\{ \sum_{i=1}^n p_i \, \ell_i(k) \;-\; V^\star(k) \Big\}.
\end{align}
The term $\text{regret}(\boldsymbol{p}, k) = \displaystyle \sum_{i=1}^n p_i \, \ell_i(k) \;-\; V^\star(k)$ is the (absolute) regret of committing to policy $\boldsymbol{p}$ when the actual budget is $k$. Assuming that the attacker's budget lies in a known discretized and finite set of scenarios \(K = \{k_1,\dots,k_m\}\), for each scenario \(k_s\), the utilities \(\ell_i(k_s)\) are computed by solving the associated attacker's subgame and exploited to compute the optimal defender's value, obtained by solving  $V^\star(k_s) = \displaystyle \max_{\boldsymbol{q}\in\Delta_n} \sum_{i=1}^n q_i \, \ell_i(k_s)$ (the defender's LP \eqref{defender_optimization_problem}). Therefore, the worst-case regret minimization problem is equivalent to the following LP~\ref{eq:mmr_lp}:
\begin{equation}\tag{WCRM-LP}\label{eq:mmr_lp}
    \begin{aligned}
        & \min_{\boldsymbol{p}\in\Delta_n,\, t\in\mathbb{R}}\quad  t \\[3pt]
        & \text{such that}\quad  V^\star(k_s) - \sum_{i=1}^n p_i \, \ell_i(k_s) \leqslant t
        \qquad s=1,\dots,m, \\[4pt]
        & \sum_{i=1}^n p_i = 1,\ \sum_{i=1}^n p_i c_i^{(\mathrm{OP})} \leqslant C_{\mathrm{OP}},\  \sum_{i=1}^n p_i \, c_i^{(\mathrm{CPU})} \leqslant C_{\text{CPU}},\ \sum_{i=1}^n p_i \, c_i^{(\mathrm{MEM})} \leqslant C_{\text{MEM}},\\
        & \sum_{i=1}^n p_i \, \tau_i \leqslant T_{\max},\quad
        \sum_{i=1}^n p_i \, r_i \geqslant R_{\min}, \quad\sum_{i\in\mathcal{F}_j} p_i \leqslant \alpha_j,\quad j=1,\dots,J
    \end{aligned}
\end{equation}

\begin{lemma}\label{lem:mmr_equiv}
    Problem \eqref{eq:mmr_lp} is equivalent to the discretized minimax regret problem
    \[
    \min_{\boldsymbol{p}\in\Delta_n} \max_{s=1,\dots,m} \Big\{V^\star(k_s) - \sum_{i=1}^n p_i \, \ell_i(k_s)\Big\}.
    \]
\end{lemma}
\begin{proof}
    For each scenario $s$, the optimum value of the defender LP when $k_s$ is known, denoted $V^\star(k_s)$, is independent of $\boldsymbol{p}$. For any strategy $\boldsymbol{p}$, the regret in the scenario \(s\) is an affine function in \(\boldsymbol{p}\). Introducing a scalar variable $t \in \mathbb{R}$ such that $t \geqslant V^\star(k_s) - \displaystyle\sum_{i=1}^n p_i \, \ell_i(k_s), \quad \forall s=1,\dots,m$ enforces that $t$ serves as an upper bound on the regret for all possible attacker budgets $k_s$. Moreover, the feasible set of defender strategies $\boldsymbol{p}$ is convex with feasibility imposed by several families of linear constraints: simplex condition $\displaystyle\sum_{i = 1}^{n} p_i = 1, p_i \geqslant 0$; resource budgets on CPU, memory, latency, quantum resilience, and family limits. All constraints are linear in $\boldsymbol{p}$, making the feasible region a polyhedron. Given any feasible $\boldsymbol{p}$, the worst-case regret over budget scenarios is $W_s = \displaystyle \max_s \Big\{V^\star(k_s) - \sum_{i=1}^n p_i \, \ell_i(k_s)\Big\}$ and the constraint in \ref{eq:mmr_lp} guarantees $t \geqslant W_s$ for all $s$, minimizing $t$ therefore has the same effect as directly minimizing $W_s$, leading to the equivalence $\displaystyle \min_{\boldsymbol{p} \in \Delta_n} \max_s \Big\{V^\star(k_s) - \sum_{i=1}^n p_i \, \ell_i(k_s)\Big\} \quad \iff \quad \min_{\boldsymbol{p}\in\Delta_n,\, t\in\mathbb{R}} t$. The objective is linear in $(\boldsymbol{p},t)$, so the robust worst-case regret problem reduces entirely to a linear program in which $t$ is minimized over a polyhedral feasible set.  
\end{proof}\noindent

\section{Performance Analysis} \label{sec:performance_analysis}

We illustrate the solution to the proposed game through two main simulation studies: the Stackelberg equilibrium under a known attacker budget and the robust minimax-regret strategy under attacker-budget uncertainty. All results are obtained by solving the relevant LPs associated with standard solvers, and the attacker subgames are solved using the \textsc{HybridAttacker} (Algorithm~\ref{algorithm_hybrid_attacker}). Table~\ref{defender_encryption_algorithms} lists several defender algorithms obtained from the literature, as well as the different parameters associated with them.

\begin{longtable}{lcccccc}
\caption{Assumed encryption algorithm parameters. The defender utility baseline values are $L_i = [85.0,\;95.0,\;92.0,\;125.0,\;140.0,\;100.0,\;105.0,\;78.0]$. Maximum budgets considered in the experimentation are as follows: $C_{\mathrm{OP}} = 8.0$, $C_{\mathrm{CPU}} = 2200000$, $C_{\mathrm{MEM}} = 1500.0$, $T_{\max} = 1200.0$, and $R_{\min} = 0.4$. The weighting coefficients are $\gamma_{\text{op}} = 0.02$, $\gamma_{\text{CPU}} = 0.00002$, $\gamma_{\text{MEM}} = 0.002$, $\gamma_{\tau} = 0.001$, and $\gamma_{r} = 0.06$. Family concentration constraints are given by $\{0{:}0.6,\;1{:}0.2,\;2{:}0.4,\;3{:}0.5\}$.}
\label{defender_encryption_algorithms}\\
\toprule
\textbf{Algorithm} & \textbf{$c_i^{(\mathrm{CPU})}$} & \textbf{$c_i^{(\mathrm{MEM})}$} & \textbf{$c_i^{\text{op}}$} & \textbf{$\tau_i$} & \textbf{$L_i$} & \textbf{$r_i$}\\
 & (cycles) & (bytes) & (rel.) & ($\mu$s) & &  \\
\midrule
\endfirsthead
\toprule
\textbf{Algorithm} & \textbf{$c_i^{(\mathrm{CPU})}$} & \textbf{$c_i^{(\mathrm{MEM})}$} & \textbf{$c_i^{\text{op}}$} & \textbf{$\tau_i$} & \textbf{$L_i$} & \textbf{$r_i$}\\
 & (cycles) & (bytes) & (rel.) & ($\mu$s) & & \\
\midrule
\endhead
\multicolumn{7}{l}{\textit{Symmetric Encryption (AEAD)}} \\
AES128-GCM & 37690 & 224 & 1.0 & 15.7 & 85 & 0.25  \\
AES-256-GCM & 41933 & 304 & 1.2 & 17.5 & 95 & 0.50 \\
ChaCha20-Poly130 & 7977 & 108 & 0.8 & 3.3 & 92 & 0.50\\
\midrule
\multicolumn{7}{l}{\textit{Post-Quantum}} \\
ML-KEM-768~\footnote{https://quarkslab.github.io/crypto-condor/2025.02.07/method/MLKEM.html} & 475000 & 4672 & 2.5 & 198 & 125 & 0.50 \\
ML-DSA-65 (Dilithium3)~\footnote{https://quarkslab.github.io/crypto-condor/devel/method/MLDSA.html} & 1950000 & 9277 & 3.0 & 811 & 140 & 0.50 \\
\midrule
\multicolumn{7}{l}{\textit{Classical Public Key}} \\
RSA-2048 & 73700000 & 1024 & 3 & 30700 & 100 & 0.0 \\
ECC P-256 & 2000000 & 256 & 2 & 825 & 105 & 0.0 \\
\midrule
\multicolumn{7}{l}{\textit{Hash Functions}} \\
SHA-256 & 10864.64 & 96 & 0.5 & 4.5 & 78 & 0.50 \\
\bottomrule
\end{longtable}

\noindent Some values used, however (especially values for $c_i^{\text{op}}$ and $L_i$), are values resulting from certain deductions in view of the data found in the references cited. The primary performance source for the cycles-per-byte measurements of our chosen encryption algorithms is the \emph{wolfSSL cryptographic benchmark}~\footnote{https://www.wolfssl.com/docs/benchmarks/}. To obtain consistent values across algorithms, the following assumptions are made: message size used for symmetric operations  $m = 1024 \text{ bytes}$ and processor frequency used to convert cycles to latency $f = 2.4 \text{ GHz}$. The latency value is estimated assuming a single cryptographic operation executed sequentially. The different values of our parameters considered are then computed as follows, $c_i^{CPU} = (\text{cycles/byte})_i \times m$ \bigg(or $c_i^{CPU} = \dfrac{f}{\text{ops/sec}_i}$ if the value is given in $ops/sec$\bigg). Latency is obtained from CPU cycles using $\tau_i = \dfrac{c_i^{CPU}}{f}$. Memory cost $c_i^{(\mathrm{MEM})}$ is derived from the algorithm specifications and corresponds to the approximate runtime memory footprint required to execute the algorithm, including key material and internal state. That said, this value varies depending on the characteristics on which each encryption algorithm is built. 

For each encryption algorithm, a list of cryptanalysis algorithms generally used and cited in the literature is given, which is represented by the following Table~\ref{attacker_cryptanalysis_algo}.  The cost parameter ($k$) represents the logarithm (base $2$) of the computational complexity stated in the literature. For instance, when Bogdanov et al. report a biclique attack with time complexity of $2^{126.1}$ operations, we extract $k = 126$ as the associated cost metric. The success probability ($s$) represents the likelihood that the attack successfully recovers the secret key when executed with the stated computational resources. For deterministic attacks, such as biclique cryptanalysis and exhaustive key search, these methods inherently ensure key recovery if completed, so we set $s = 0.99$.
In the case of probabilistic attacks (e.g., side-channel, fault injection, or quantum approaches involving measurement uncertainty), we derive success probabilities from empirical success rates reported in the literature, taking into account factors like measurement noise, fault injection effectiveness, and quantum decoherence. When the original papers do not provide explicit success probabilities, we estimate them based on the attack’s data complexity and the statistical confidence achievable with the given number of samples or traces. Overall, the values presented here are not directly stated in the referenced works; instead, they are approximations grounded in the characteristics described by the authors, aiming to reflect as much as possible realistic behavior across these attack types.

{\footnotesize{
\begin{longtable}{l l c c l}

\caption{Attack methods for each encryption algorithm}
\label{attacker_cryptanalysis_algo} \\

\toprule
\textbf{Target Algorithm} & \textbf{Attack Method} & \textbf{Prob. ($s$)} & \textbf{Cost ($k$)} & \textbf{Reference} \\
\midrule
\endfirsthead

\multicolumn{5}{l}{\textit{\textbf{AES128-GCM}}} \\
& Biclique Cryptanalysis & 0.99 & 126 & \cite{hutchison_biclique_2011} \\
& Improved biclique cryptanalysis & 0.99 & 126 & \cite{foo_improving_2015}\\
& Quantum Grover Search & 0.92 & 64 & \cite{takagi_applying_2016} \\
& Differential Fault Analysis & 0.85 & 25 & \cite{goos_differential_2003} \\
& Electromagnetic Side-Channel Attack & 0.8 & 13 & \cite{wang_far_2020} \\
& Exhaustive key search & 0.99 & 128 & \cite{takagi_applying_2016} \\

\midrule
\multicolumn{5}{l}{\textit{\textbf{AES-256-GCM}}} \\
& Biclique Cryptanalysis & 0.99 & 254 & \cite{hutchison_biclique_2011} \\
& Quantum Grover Search & 0.99 & 128 & \cite{takagi_applying_2016} \\
& Related-key boomerang & 0.95 & 100 & \cite{hutchison_related-key_2009} \\
& Cache‑Based Side‑Channel Attack & 0.88 & 15 & \cite{neve_complexity_2007} \\
& Correlation Power Analysis (CPA) & 0.99 & 12 & \cite{mestiri_evaluating_2025} \\

\midrule
\multicolumn{5}{l}{\textit{\textbf{ChaCha20-Poly1305}}} \\
& PNB-focused differential attack (7-round) & 0.5 & 255.62 & \cite{miyashita_pnb-focused_2021} \\
& Differential-Linear Cryptanalysis (7-round) & 0.12 & 189.7 & \cite{xu_differential-linear_2024} \\
& Differential-Linear Cryptanalysis (7.25-round) & 0.12 & 223.9 & \cite{xu_differential-linear_2024} \\
& Extension of PNBs distinguisher (7.25-round) & 0.08 & 228.24 & \cite{dey_advancing_2024} \\
& Extension of PNBs distinguisher (7.5-round) & 0.08 & 255.24 & \cite{dey_advancing_2024} \\
& Fault Injection & 0.95 & 8 & \cite{dilip_kumar_practical_2017} \\

\midrule
\multicolumn{5}{l}{\textit{\textbf{ML-KEM-768 (Kyber)}}} \\
& Lattice Reduction (BKZ) & 0.99 & 161 & \cite{bos2018crystals} \\
& Hybrid Attack & 0.08 & 150 & \cite{gopfert_hybrid_2017} \\
& Side-Channel (Non-pro CPA) & 0.50 & 60 & \cite{wang_far_2020} \\
& Timing Attack (KyberSlash2) & 0.99 & 25 & \cite{bernstein2011ebacs} \\
& Timing Attack (KyberSlash1) & 0.99 & 20.8 & \cite{bernstein2011ebacs} \\
& Brute Force (Quantum) & 0.99 & 192 & \cite{national_institute_of_standards_and_technology_us_secure_2015} \\

\midrule
\multicolumn{5}{l}{\textit{\textbf{ML-DSA-65 (Dilithium)}}} \\
& Lattice Reduction (BKZ) & 0.05 & 170 & \cite{national_institute_of_standards_and_technology_us_advanced_2023} \\
& Side-Channel (EM) & 0.7 & 12 & \cite{marzougui_profiling_2022}\\
& Side-Channel (Power) & 0.8 & 13 & \cite{wang_far_2020} \\
& Side‑Channel via Rejected Signatures & 0.8 & 30 & \cite{wang_far_2020} \\
& Signature Correction Attack & 0.65 & 16 & \cite{islam_signature_2022} \\

\midrule
\multicolumn{5}{l}{\textit{\textbf{RSA-2048}}} \\
& GNFS (Classical) & 0.99 & 112 & \cite{kleinjung_factorization_2010}\\
& Timing Attack & 0.4 & 12 & \cite{kocher_timing_1996} \\
& Differential Power Analysis & 0.45 & 10 & \cite{kocher_differential_1999} \\
& CRT Fault Attack & 0.50 & 8 & \cite{kocher_differential_1999} \\

\midrule
\multicolumn{5}{l}{\textit{\textbf{ECC P-256}}} \\
& Pollard's rho & 0.99 & 128 & \cite{pollard_monte_1978} \\
& Parallel Pollard’s Rho & 0.99 & 126 & \cite{chari_towards_1999} \\
& Differential Power Analysis & 0.45 & 10 & \cite{kocher_differential_1999} \\

\midrule
\multicolumn{5}{l}{\textit{\textbf{SHA-256}}} \\
& Preimage Brute Force & 0.99 & 256 & \cite{national_institute_of_standards_and_technology_us_advanced_2023} \\
& Grover's (Quantum) & 0.99 & 128 & \cite{shor1997} \\
& Side‑Channel Attacks & 0.40 & 10 & \cite{kocher_timing_1996} \\
\bottomrule
\end{longtable} }}

\subsection{Stackelberg Equilibrium: Known Budget ($ k = 40$)} \label{subsec:base_equilibrium}

Table~\ref{tab:final-results} presents the Stackelberg equilibrium obtained from the defender optimization when the attacker has a budget $k = 40$, an attack value $v = 300$, and a linear cost function $\varphi(x) = x$, computed under operational, CPU, memory, latency, resilience, and family-diversification constraints given in Table~\ref{defender_encryption_algorithms}.

\begin{table}[h!]
\centering
\small
\caption{Stackelberg equilibrium results under multi-constraint defender model ($k = 40$).}
\label{tab:final-results}
\resizebox{\textwidth}{!}{%
\begin{tabular}{l c c c c c p{4.5cm} c c}
\toprule
\textbf{Algorithm} & \textbf{$p_i$} & \textbf{\%} & \textbf{$\ell_i$} &
\textbf{op} & \textbf{family} & \textbf{Attacker Strategy} &
\textbf{Breach prob.} & \textbf{Global Metrics} \\
\midrule
AES-128-GCM            & 0.0000 & 0.00\%  & 1.33  & 1.0 & 1 &
DFA, EM Side-Channel & 0.970 & \\[1ex]

AES-256-GCM            & 0.0000 & 0.00\%  & -0.51 & 1.2 & 1 &
CPA & 0.990 & \\[1ex]

ChaCha20-Poly1305      & 0.2000 & 20.00\% & 4.24  & 0.8 & 1 &
Fault Injection & 0.950 & \\[1ex]

ML-KEM-768             & 0.2000 & 20.00\% & -17.81 & 2.5 & 3 &
KyberSlash1 & 0.990 & \\[1ex]

ML-DSA-65              & 0.0000 & 0.00\%  & -49.99 & 3.0 & 3 &
EM Side-Channel, Power Side-Channel & 0.940 & \\[1ex]

RSA-2048               & 0.0000 & 0.00\%  & -1490.31 & 3.0 & 0 &
Timing Attack, DPA, CRT Fault & 0.835 & \\[1ex]

ECC-P256               & 0.2000 & 20.00\% & 16.37 & 2.0 & 0 &
DPA & 0.450 & \\[1ex]

SHA-256                & 0.4000 & 40.00\% & 46.41 & 0.5 & 2 &
Side-Channel & 0.400 &
\textbf{Exp. Gain: 19.12} \\[1ex]

\bottomrule
\end{tabular}%
}
\end{table}

\noindent The Stackelberg equilibrium results reveal that the defender adopts a
\emph{mixed strategy} over four algorithms: SHA-256, ChaCha20-Poly1305, ML-KEM-768, and ECC-P256. Among these, SHA-256 receives the largest probability mass ($40\%$), reflecting its high utility score ($\ell_i = 46.41$) combined with very low operational and computational costs. ChaCha20-Poly1305, ML-KEM-768, and ECC-P256 each receive a probability of $20\%$, ensuring diversity across algorithm families while respecting the defender's system constraints. Several algorithms are excluded from the equilibrium support. AES128-GCM and AES-256-GCM receive zero probability, likely due to their low expected utilities while their attacker success probabilities are high under the optimal attack strategy. Similarly, ML-DSA-65 and RSA-2048 are excluded regarding their excessive resource requirements and unfavorable utility scores. In particular, RSA-2048 exhibits a strongly negative utility value ($\ell_i = -1490.31$), reflecting the large computational overhead combined with limited resilience benefits in the modeled environment.
The resulting defender strategy satisfies all system constraints with the following values: operational cost: 1.2600 / 8.0; CPU usage: 500941.2560 / 2200000.0; Memory usage: 1045.6000 / 1500.0; latency: 207.0600 / 1200.0; resilience:  0.4000 and \emph{expected objective (gain): 19.1217}. From a security perspective, the equilibrium produces an expected breach probability of approximately $0.638$ against an optimal attacker response. Although this probability may appear relatively high in isolation, it reflects the realistic assumption that the attacker can adaptively select the most profitable attack vector subject to their resource constraints. The equilibrium therefore represents the defender's best possible strategy under worst-case attacker behavior.

These results highlight the practical importance of the proposed Stackelberg optimization approach, which enables the defender to randomize across multiple algorithm families strategically. Following the equilibrium distribution then balances performance constraints and security robustness, limiting the attacker's ability to focus resources on a single vulnerability and forcing them to distribute efforts across different cryptanalytic methods. In real-world deployments, such a strategy can be interpreted as a dynamic cryptographic configuration policy where systems periodically rotate or probabilistically select among several approved primitives. For instance, with the above parameters, a secure communication infrastructure could deploy a mixture of symmetric encryption (ChaCha20-Poly1305), post-quantum key encapsulation (ML-KEM-768), classical elliptic-curve cryptography (ECC-P256), and hashing mechanisms (SHA-256). To further highlight the importance of the proposed optimal strategy, we compare it with some heuristics in which the defender chooses his strategy randomly while the attacker responds optimally. Figure~\ref{optimal_vs_random} illustrates this comparison, showing that the defender's gain is lower when they deviate from the optimal recommendation.
\begin{figure}[ht!]
    \centering
    \caption{\textbf{Relevance of the proposed optimal defender strategy.}
    The figure shows that any deviation from the optimal strategy leads to a reduction in the defender's utility. This degradation is particularly significant for purely random strategies. In contrast, strategies that partially incorporate the proposed optimization framework (by enforcing only a subset of the constraints) achieve improved performance, but still incur a noticeable loss compared to the fully optimal solution.}
    \label{optimal_vs_random}
    \includegraphics[width=1\linewidth]{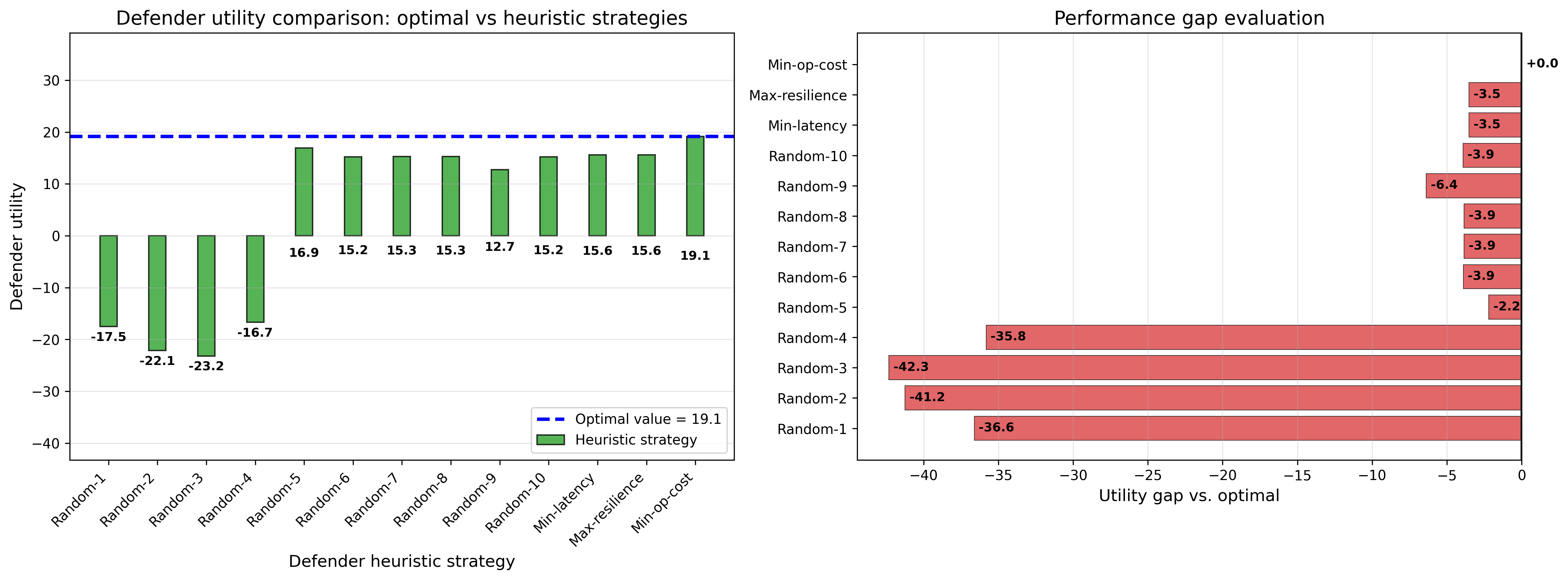}
\end{figure}

\noindent The random heuristic strategies shown in Figure~\ref{optimal_vs_random} are generated by solving an LP with randomly sampled objective coefficients. For each instance, we sample a coefficient vector $c \in \mathbb{R}^n$ from an independently and identically distributed standard normal distribution $\mathcal{N}(0,1)$ (through the function \emph{np.random.randn(n)}),  where $n$ is the number of candidate algorithms. We then solve $\min c^\top p$ subject to the same resource constraints used for the optimal Stackelberg strategy. Because the objective is linear, the LP solver returns an extreme point of the feasible polytope, a vertex selected according to a random linear direction. The coefficients are drawn from a Gaussian distribution to avoid systematic bias toward any particular region of the feasible space. All generated strategies are tested in terms of constraint feasibility before inclusion in the comparison. The min-opt-cost, min-latency, and max-resilience heuristics are obtained by optimizing operational cost, latency, and quantum resilience, respectively, using the same constraint set.

\subsection{Robust Minimax-Regret Strategy: Uncertain Budget} \label{subsec:robust_results}

We now evaluate the robust formulation from Section~\ref{sec:worst_case_regret_minimization} under different attacker-budget scenarios. Unlike the Stackelberg model, which assumes a known adversarial budget, here the defender accounts for uncertainty in the attacker's resources and aims for a strategy that performs reasonably well across a range of possible budgets. We consider attacker-budget scenarios $k \in \{11, 15, 20, 25, 30\}$ and, for each scenario, we compute the attacker's optimal subset of attack methods and derive the corresponding defender utilities $\ell_i(k)$.

\paragraph{Scenario utilities and optimal defender values}

\begin{align*}
k = 11: &\quad V^\star = 70.154, \\
&\boldsymbol{\ell}(k) =
[11.5275,\;93.5418,\;4.2352,\;105.938,\;81.605,\;-1456.808,\;16.373,\;46.406] \\
k = 15: &\quad V^\star = 54.054, \\
&\boldsymbol{\ell}(k) =
[11.5275,\;-0.5082,\;4.2352,\;105.938,\;-30.395,\;-1456.808,\;16.373,\;46.406] \\
k = 20: &\quad V^\star = 53.997, \\
&\boldsymbol{\ell}(k) =
[1.3275,\;-0.5082,\;4.2352,\;105.938,\;-30.395,\;-1479.308,\;16.373,\;46.406] \\
k = 25: &\quad V^\star = 19.122, \\
&\boldsymbol{\ell}(k) =
[1.3275,\;-0.5082,\;4.2352,\;-17.812,\;-49.995,\;-1479.308,\;16.373,\;46.406] \\
k = 30: &\quad V^\star = 19.122, \\
&\boldsymbol{\ell}(k) =
[1.3275,\;-0.5082,\;4.2352,\;-17.812,\;-49.995,\;-1490.308,\;16.373,\;46.406]
\end{align*}

\paragraph{Minimax-regret LP}

The robust defender strategy is computed by solving the~\ref{eq:mmr_lp}, the optimal solution yields the minimax-regret strategy $\boldsymbol{p^{\mathrm{mmr}}}$, shown in Table~\ref{regret_values}.

\begin{table}[ht!]
\centering
\caption{Robust minimax-regret strategy $\boldsymbol{p^{\mathrm{mmr}}}$.}
\label{regret_values}
\begin{tabular}{p{5cm}p{4cm}p{3cm}}
\toprule
Algorithm & Probability $\boldsymbol{p^{\mathrm{mmr}}}$ & Percentage \\
\midrule
AES-128-GCM        & 0.000000 & 0.00\% \\
AES-256-GCM        & 0.171190 & 17.12\% \\
ChaCha20-Poly1305  & 0.000000 & 0.00\% \\
ML-KEM-768         & 0.282278 & 28.23\% \\
ML-DSA-65          & 0.000000 & 0.00\% \\
RSA-2048           & 0.000000 & 0.00\% \\
ECC-P256           & 0.146531 & 14.65\% \\
SHA-256            & 0.400000 & 40.00\% \\
\bottomrule
\end{tabular}
\end{table}

The regret values $V^\star(k) - \displaystyle\sum_i p_i^{\mathrm{mmr}} \ell_i(k)$ associated with each scenario are,
\begin{itemize}
\item $k=11$: regret $=3.2750$ ; \quad $k=15$: regret $=3.2750$ ; \quad $k=20$: regret $=3.2188$
\item $k=25$: regret $=3.2750$ ; \quad $k=30$: regret $=3.2750$
\end{itemize}
\noindent Table~\ref{tab:regret_matrix} compares the worst-case regret of each scenario-specific optimal strategy with the minimax-regret strategy.

\begin{table}[H]
\centering
\caption{Regret comparison matrix across strategies and scenarios, ${}^*$Optimal for this scenario. Opt$(k = z)$ suppose that the defender adopts the optimal strategy associated with the budget $z$ against another budget.}
\label{tab:regret_matrix}
\begin{tabular}{lcccccc}
\toprule
\multirow{2}{*}{Scenario} &
\multicolumn{5}{c}{Actual budget $k$} &
\multirow{2}{*}{Max Regret} \\
& 11 & 15 & 20 & 25 & 30 & \\
\midrule
Opt$(k=11)$ & $0.000^*$ & 2.7095 & 2.6533 & 4.1628 & 4.1628 & 4.1628 \\
Opt$(k=15)$ & 16.1005 & $0.000^*$ & 1.9381 & 3.9105 & 3.9105 & 16.1005 \\
Opt$(k=20)$ & 16.1567 & 0.0562 & $0.000^*$ & 2.1951 & 2.1951 & 16.1567 \\
Opt$(k=25)$ & 26.2824 & 10.1819 & 10.1257 & $0.000^*$ & 0.0000 & 26.2824 \\
Opt$(k=30)$ & 26.2824 & 10.1819 & 10.1257 & 0.0000 & $0.000^*$ & 26.2824 \\
\midrule
$\boldsymbol{p^{\mathrm{mmr}}}$ & 3.2750 & 3.2750 & 3.2188 & 3.2750 & 3.2750 & \textbf{3.2750} \\
\bottomrule
\end{tabular}
\end{table}

\noindent The results show that strategies designed for a single attacker budget can struggle when the actual attacker capabilities differ. Consider the strategy optimized for $k=25$: it works perfectly when the attacker budget really is $25$ (zero regret), but performs poorly, with a regret of $26.2824$, when the budget turns out to be smaller. The other scenario-specific strategies behave similarly. This suggests that assuming a fixed attacker budget and optimizing accordingly can lead to brittle security configurations.

The minimax-regret strategy $\boldsymbol{p^{\mathrm{mmr}}}$ takes a different approach by balancing performance across all scenarios. Its worst-case regret is $\mathbf{3.275}$, which outperforms the best single-scenario strategy (Opt$(k=11)$ with max regret $4.1628$) and is far better than the worst ones (Opt$(k=25)$ with max regret $26.2824$), highlighting the value of robust optimization when we are unsure about the attacker's actual resources. Interestingly, the minimax-regret strategy assigns probability mass to four algorithms: SHA-256 ($40\%$), ML-KEM-768 ($28.23\%$), AES-256-GCM ($17.12\%$), and ECC-P256 ($14.65\%$). SHA-256 receives the largest weight, reflecting its favorable utility values and consistent performance across scenarios. ML-KEM-768 follows, bringing post-quantum resilience into the mix. AES-256-GCM and ECC-P256 receive moderate weights, enough to maintain diversification without adding much risk.

\paragraph{Breach probability robustness}

A complementary analysis is performed using breach-probability regret, and the optimal breach probabilities (computed by taking the optimal strategy of the defender for each attacker budget scenario) are $B^\star = [0.2077,\;0.3454,\;0.3454,\;0.6331,\;0.6331].$
The minimax-regret strategy achieves a worst-case breach regret of $\mathbf{0.0500}$, outperforming all scenario-specific strategies as shown in Table~\ref{breach_probability} below.

\begin{table}[!ht]
\centering
\caption{Breach-probability regret: comparison of strategies across scenarios}
\label{breach_probability}
\begin{tabular}{lcccccc}
\toprule
\multirow{2}{*}{Scenario} &
\multicolumn{5}{c}{Actual budget $k$} &
\multirow{2}{*}{Max Regret} \\
& 11 & 15 & 20 & 25 & 30 & \\
\midrule
Opt(K=11.0) & 0.0000$^{*}$ & 0.0603 & 0.0603 & 0.0637 & 0.0637 & 0.0637 \\
Opt(K=15.0) & 0.1688 & 0.0000$^{*}$ & 0.0545 & 0.0616 & 0.0616 & 0.1688 \\
Opt(K=20.0) & 0.1377 & 0.0000 & 0.0000$^{*}$ & 0.0089 & 0.0089 & 0.1377 \\
Opt(K=25.0) & 0.2323 & 0.0946 & 0.0946 & 0.0000$^{*}$ & 0.0049 & 0.2323 \\
Opt(K=30.0) & 0.2323 & 0.0946 & 0.0946 & 0.0049 & 0.0000$^{*}$ & 0.2323 \\
\midrule
MMR $\boldsymbol{p^{\mathrm{mmr}}}$ & 0.0182 & 0.0500 & 0.0500 & 0.0418 & 0.0418 & \textbf{0.0500} \\
\bottomrule
\end{tabular}
\end{table}

\noindent For comparison, the best scenario-specific strategy (Opt$(k=11)$) has a worst-case breach regret of $0.0623$, while the worst-performing strategy (Opt$(k=25)$) reaches $0.2323$, confirming that the robust defender strategy not only minimizes utility regret but also maintains stable security performance in terms of expected breach probability.

\paragraph{Practical interpretation} From a deployment perspective, the minimax-regret strategy can be interpreted as a probabilistic cryptographic configuration policy that the system rotates among multiple algorithms drawn from symmetric, post-quantum, elliptic-curve, and hashing families, rather than committing to a single cryptographic primitive, thereby reducing the effectiveness of targeted attacks. And even if the attacker has more resources than we expected, performance doesn't crash, it just degrades more smoothly. So robust optimization gives us a practical way to build hybrid cryptographic systems that handle uncertainty about the attacker's capabilities.
		
\section{Conslusion} \label{sec:conclusion}	

In this work, we presented a Stackelberg game-theoretic model for cryptographic algorithm hybridization. Our approach enables defenders to probabilistically select among multiple encryption algorithms, thereby enhancing overall system resilience and reducing an attacker's probability of successful cryptanalysis. We formulated the defender-attacker interaction as a two-stage game where the attacker, operating under budget constraints, observes the defender's strategy commitment before selecting their optimal attack portfolio. The attacker's subgame is characterized as a non-monotone submodular maximization problem subject to knapsack constraints, a problem class typically addressed through approximation algorithms in the literature. However, our problem exhibits structural properties, specifically, the relatively small number of available cryptanalysis methods per algorithm, combined with the explicit functional form of the attacker's utility function, allowing us to employ DP to obtain exact optimal solutions. This stands in contrast to existing approximation algorithms (such as \textsc{SampleGreedy}), which are designed for large-scale instances where dynamic programming becomes computationally intractable, but which only guarantee bounded approximation ratios rather than exact optimality. For the defender's problem, we formulated a multi-constraint linear programming model that accounts for operational costs, computational resources, latency requirements, and quantum resilience thresholds, while enforcing cryptographic family diversification constraints. Recognizing that the defender may lack precise knowledge of the attacker's budget, we extended our framework to incorporate worst-case regret minimization, ensuring robust performance across a range of attacker capabilities. Our experimental evaluation demonstrates the effectiveness of the proposed optimization-based approach compared to heuristic strategies. The comparison against feasible baseline strategies, including latency-minimization and quantum-resilience-maximization, reveals that the optimal mixed strategy achieves superior performance under these constraints, validating the relevance of our principled optimization framework. The model, however, remains limited by its static Stackelberg structure and the independence of cryptanalysis methods, which overlook synergies, shared precomputation, and implementation‑level attacks. The linear formulation of defender costs also abstracts away nonlinear or protocol‑level constraints that could exist in real deployments. Extending the framework to Bayesian settings, dynamic or repeated interactions, and protocol‑aware hybridization would broaden its applicability and capture more realistic adversarial behavior.

\appendix \label{appendix}

\section{Execution trace: \textsc{SampleGreedy} vs.\ DP}
\label{app:samplegreedy}

\noindent\textbf{Instance:} $v=1200$, $K=500$, $\phi(x)=x$, $q=0.414$, seed $= 65$.

\smallskip
\noindent\textbf{Phase 1 — Best single method.}
Individual utilities: $A_1{=}140$, $A_2{=}220$, $A_3{=}224$, $A_4{=}180$. 
Best single: $A_3$ (utility 224.0).

\smallskip
\noindent\textbf{Phase 2 — Greedy construction.}

\begin{itemize}
  \item \textbf{It.\ 1:} $S=\emptyset$, $R=500$. Densities $ = (v \times s_{i,j} - k_j) / k_j$: $A_4 = 1.50$ (best), $A_1 = 1.40$, $A_2 = 1.10$, $A_3{=}0.80$. Coin flip: \textsc{reject} $A_4$.
  \item \textbf{It.\ 2:} $S=\emptyset$, $R=500$. Best density: $A_1{=}1.40$. Coin flip: \textsc{accept} $A_1$. $\Rightarrow S=\{A_1\}$, $R=400$,\ the marginal gain of adding $A_2$ and $A_3$ are $\Delta P_{\mathrm{succ}}(A_2 \mid \{A_1\}) = 1 - (0.8)(0.65) - 0.2 = 0.28$ and $\Delta P_{\mathrm{succ}}(A_3 \mid \{A_1\}) = 1 - (0.8)(0.58) - 0.2 = 0.336$ respectively.
  \item \textbf{It.\ 3:} Best density: $A_2 = (1200 \times 0.28 - 200) / 200 = 0.68$. Coin flip: \textsc{reject} $A_2$.
  \item \textbf{It.\ 4:} Best density: $A_3 = (1200 \times 0.336 - 280) / 280 = 0.44$. Coin flip: \textsc{reject} $A_3$.
  \item \textbf{It.\ 5:} No feasible candidates. \textsc{Terminate}.
\end{itemize}

\noindent Greedy output: $\{A_1\}$, $F=140.0 <$ best single.
\textsc{SampleGreedy} returns $\mathbf{\{p_3\}}$, $F^{\textsc{SG}}=\mathbf{224.0}$.

\smallskip
\noindent\textbf{Phase 3 — Exact DP.}
Optimal solution: $S^*_{\mathrm{DP}} = \{A_1, A_2, A_4\}$, cost $= 420 \leqslant 500$,
$P_{\mathrm{succ}} = 0.61$, $F^{\mathrm{DP}} = \mathbf{312.0}$.

\smallskip
\noindent\textbf{Optimality gap:} $(312 - 224)/312 \approx \mathbf{28.2\%}$.
This confirms that \textsc{SampleGreedy}, due to its randomized acceptance, 
may miss the optimal combination even on small instances, whereas DP guarantees 
the exact optimum at negligible computational cost for such problem sizes.

\bibliographystyle{plain}
\bibliography{references}

\end{document}

%% file: preamble.tex
\usepackage{amsmath,amssymb,amsthm}
\usepackage{booktabs}
\usepackage{hyperref}
\usepackage{graphicx}
\usepackage{listings}
\usepackage{float}
\usepackage{color}
\usepackage{tabularx}
\newtheorem{theorem}{Theorem}

\usepackage{xcolor}
\newtheorem{remark}{Remark}

\newtheorem{lemma}{Lemma}

\usepackage{algorithm}      % provides the algorithm float
\usepackage{algpseudocode}  % provides algorithmic environment with \State,\For,\If,\Return,...
\usepackage{multirow}
\usepackage{array}
\usepackage{amsfonts}
\usepackage{caption}
\usepackage{hyperref}
\usepackage{subcaption}
\usepackage{geometry}
\newcommand{\Input}{\item[\textbf{Input:}]}
\newcommand{\Output}{\item[\textbf{Output:}]}

\usepackage{todonotes}

\newenvironment{breakablealgorithm}
  {
   \begin{center}
     \refstepcounter{algorithm}
     \hrule height .8pt depth 0pt \kern 2pt
     \renewcommand{\caption}[2][\relax]{%
       {\raggedright \textbf{Algorithm~\thealgorithm} ##2\par}%
       \kern 2pt\hrule\kern 2pt
     }
  }
  {
     \kern 2pt\hrule\relax
   \end{center}
  }

\usepackage{longtable}
\usepackage{booktabs} % pour \toprule etc.

\DeclareMathOperator*{\argmax}{arg\,max}